\pdfoutput=1	

\documentclass[11pt]{article}

\usepackage{fullpage}
\usepackage{booktabs} 
\usepackage{graphicx}
\usepackage{balance}  
\usepackage{xspace}
\usepackage{subcaption}
\usepackage{color}
\usepackage{enumitem}
\usepackage[linesnumbered, ruled, vlined]{algorithm2e}
\usepackage{quotes}
\usepackage{amsmath,amsthm,amssymb}
\usepackage{listings}
\usepackage{multirow}
\usepackage{float}

\newcommand{\var}{\mathrm{Var}}
\newcommand{\prob}[1]{{\mathrm{Pr}} \left [ #1 \right ]}
\newcommand{\expec}[1]{{\mathbb E}\left [ #1 \right ]}
\newcommand{\std}[1]{{\mathbb S}\left [ #1 \right ]}
\newcommand{\variance}[1]{{\mathbb{VAR}}\left [ #1 \right ]}
\newcommand{\cv}[1]{{\mathbb{CV}}\left [ #1 \right ]}
\newcommand{\error}{\mathrm{error}}

\newtheorem{theorem}{Theorem}

\newtheorem{lemma}{Lemma}

\floatstyle{ruled}
\newfloat{query}{b}{lok}
\floatname{query}{Query~AQ}
\newfloat{query2}{b}{lok}
\floatname{query2}{Query~B}
\newcommand{\todo}[1]{\textcolor{blue}{}\PackageWarning{TODO:}{#1!}} 
\newcommand{\remove}[1]{}	    
\newcommand{\main}[1]{}				
\newcommand{\arxiv}{\iftrue}		 

\definecolor{red}{rgb}{0,0,0}		
\definecolor{blue}{rgb}{0,0,0}		

\newcommand{\graphwidth}{0.6\columnwidth}

\newcommand{\eg}{\hbox{\emph{e.g.,}}\xspace}
\newcommand{\ie}{\hbox{\emph{i.e.,}}\xspace}

\newcommand{\etc}{\hbox{\emph{etc.}}\xspace}

\newcommand{\uniform}{\texttt{Uniform}\xspace}		
\newcommand{\senate}{\texttt{SENATE}\xspace}		

\newcommand{\sasg}{\texttt{SASG}\xspace}	
\newcommand{\samg}{\texttt{SAMG}\xspace}	
\newcommand{\masg}{\texttt{MASG}\xspace}	
\newcommand{\mamg}{\texttt{MAMG}\xspace}	

\newcommand{\cvopt}{\texttt{CVOPT}\xspace}
\newcommand{\cvoptinf}{\texttt{CVOPT-INF}\xspace}
\newcommand{\cvoptsasg}{\texttt{CVOPT-SASG}}

\newcommand{\cs}{\texttt{CS}\xspace}
\newcommand{\rl}{\texttt{RL}\xspace}

\newcommand{\sps}{\texttt{Sample+Seek}\xspace}

\newcommand{\openaq}{\texttt{OpenAQ}\xspace}
\newcommand{\openaqtb}{\texttt{OpenAQ-25x}\xspace}
\newcommand{\bikes}{\texttt{Bikes}\xspace}

\begin{document}
	
	\title{Random Sampling for Group-By Queries}

	\makeatletter
	\newcommand{\linebreakand}{%
	\end{@IEEEauthorhalign}
	\hfill\mbox{}\par
	\mbox{}\hfill\begin{@IEEEauthorhalign}
	}
	\makeatother
	
	\author{Trong Duc Nguyen$^1$ \and Ming-Hung Shih$^1$ \and 
Sai Sree Parvathaneni$^1$
 \and Bojian Xu$^2$ \and Divesh
  Srivastava$^3$ \and Srikanta Tirthapura$^1$}

\date{%
    $^1$Iowa State
    University, IA 50011, USA\\%
    $^2$ Eastern Washington     University, WA 99004, USA.\\
    $^3$AT\&T Labs--Research, NJ 07921, USA\\
\ \\
   \texttt{trong@iastate.edu\ \  mshih@iastate.edu\ \ 
     ssree@iastate.edu \ \ 
bojianxu@ewu.edu \ \ 
 divesh@research.att.com\ \ snt@iastate.edu}\\
\ \\
    \today
}

	\maketitle
	
	\begin{abstract}
		Random sampling has been widely used in approximate query processing on large databases, due to its potential to significantly reduce resource usage and response times, at the cost of a small approximation error. We consider random sampling for answering the ubiquitous class of group-by queries, which first group data according to one or more attributes, and then aggregate within each group after filtering through a predicate. The challenge with group-by queries is that a sampling method cannot focus on optimizing the quality of a single answer (e.g. the mean of selected data), but must simultaneously optimize the quality of a set of answers (one per group). 
		
		We present \cvopt, a query- and data-driven sampling framework for a set of queries that return multiple answers, e.g. group-by queries. To evaluate the quality of a sample, \cvopt defines a metric based on the norm (e.g. $\ell_2$ or $\ell_\infty$) of the coefficients of variation (CVs) of different answers, and constructs a stratified sample that provably optimizes the metric. \cvopt can handle group-by queries on data where groups have vastly different statistical characteristics, such as frequencies, means, or variances. \cvopt jointly optimizes for multiple aggregations and multiple group-by clauses, and provides a way to prioritize specific groups or aggregates. It can be tuned to cases when partial information about a query workload is known, such as a data warehouse where queries are scheduled to run periodically.  
		
		Our experimental results show that \cvopt outperforms current state-of-the-art on sample quality and estimation accuracy for group-by queries. On a set of queries on two real-world data sets, \cvopt yields relative errors that are $5 \times$ smaller than competing approaches, under the same budget.
	\end{abstract}
\section{Introduction}
\label{sec:intro}
As data size increases faster than the computational resources for query processing, answering queries on large data with a reasonable turnaround time has remained a significant challenge. One solution to handling this data deluge is through random sampling. A sample is a subset of data, collected with the use of randomness to determine which items are included in the sample and which are not. A query posed on the data can be quickly and approximately answered by executing the query on the sample, followed by an appropriate normalization. Sampling is attractive when a good trade-off can be obtained between the size of the sample and the accuracy of the answer.

We investigate random sampling to support ``group-by'' queries. A group-by query partitions an input relation into multiple groups according to the values of one or more attributes, and applies an aggregate function within each group. For instance, on a relation consisting of student records {\tt Student(name, id, year, major, gpa)}, the following SQL query asks for the average {\tt gpa} of each major: 
\begin{lstlisting}
	SELECT   major, AVG(gpa) 
	FROM     Student 
	GROUP BY major
\end{lstlisting}
Group-by queries are very common, especially in data warehouses. For instance, in the TPC-DS benchmark~\cite{tpcds}, 81/ 99 queries have group-by clause. We address the basic question: how to {\color{blue}sample a table} to accurately answer group-by queries?

A simple method is {\em uniform sampling}, where each element is sampled with the same probability. It has been recognized that uniform sampling has significant shortcomings~\cite{CDMN01}. Since a group will be represented in proportion to its volume (i.e. the number of elements in the group), a group whose volume is small may have very few elements selected into the sample or may be missing altogether, while a group whose volume is large may have a large number of elements in the sample. This can clearly lead to high errors for some (perhaps a large fraction of) groups, as confirmed by our experimental study. 

The accuracy of group-by queries can be improved using stratified sampling~\cite{AGP00, CDN-TODS2007,RL-EDBT2009}. Data is partitioned into multiple strata. Uniform sampling is applied within each stratum, but different probabilities may be used across strata. The ``best'' way to use stratified sampling for group-by queries is not obvious -- how to stratify the data into groups, and how to assign a sampling probability to a stratum?


Prior work on ``congressional sampling''~\cite{AGP00} has advocated the use of a fixed allocation of samples to each stratum, irrespective of its size (the ``senate'' strategy)\footnote{More accurately, a hybrid of fixed and population-based allocations.}. Consider the above example, where the table is grouped by the attribute {\tt major}. Suppose the population was stratified so that there is one stratum for each possible value of {\tt major}, and uniform sampling is applied within each stratum. Even for this simple example, congressional sampling can lead to a sub-optimal allocation of samples. Consider two groups, $1$ and $2$ with the same number of elements and the same mean, but group $1$ has a higher variance than group $2$. In estimating the mean of each stratum, allocating equal numbers of samples to each group leads to a worse error for group $1$ than for group $2$. It is intuitively better to give more samples to group $1$ than to group $2$, but it is not clear how many more. {\em Thus far, there has been no systematic study on sample allocation to groups, and our work aims to fill this gap.}


Our high-level approach is as follows. We first define a metric that quantifies the error of an estimate for a group-by query through using the sample, and we obtain an allocation and then a sample that minimizes this metric. Such an optimization approach to sampling has been used before. For instance, the popular ``Neyman allocation''~\cite{Neyman1934,edbt19} assigns samples to strata to {\em minimize the variance of an estimate of the population mean} that can be derived from the sample. The challenge with group-by queries is that there isn't a single estimate to focus on, such as the population mean, but instead, an answer for each group.

Since multiple estimates are desired, we combine the metrics corresponding to different estimates to derive a global metric for the quality of a sample. One possible global metric is the sum of the variances of the different estimators. Such a global metric is inadequate in the case when the means of different groups are vastly different, in which case variances cannot be compared across groups. For instance, consider 2 groups with means $\mu_1 = 1000$ and $\mu_2 = 100$. Suppose an estimator for each group has a variance of $100$. Though both estimators have the same variance, they are of different quality. If data within each group follows a normal distribution, then we can derive that with a probability greater than $96\%$, the estimated mean $y_1$ is between $\left[{950; 1050}\right]$ and the estimated mean $y_2$ is between $\left[50; 150\right]$.  In other words, $y_1 = \mu_1 \pm 5\%$ and $y_2 = \mu_2 \pm 50\%$, and therefore $y_1$ is a much better estimate than $y_2$, in terms of relative accuracy. However, since the global metric adds up the variances, both the estimators contribute equally to the optimization metric. This leads to an allocation that favors groups that have a large mean over groups that have a small mean.

\smallskip

\noindent {\bf Coefficient of Variation (CV):} In order to combine the accuracies of different estimators, perhaps with vastly different expectations, it is better to use the {\em coefficient of variation} (also called "relative variance"). 
The coefficient of variation of a random variable $X$ is defined as $\cv{X} = \frac{\std{X}}{\expec{X}}$, where $\expec{X}$ is the expectation (mean) and $\std{X}$ the standard deviation of $X$, respectively. We assume that the attribute that is being aggregated has a non-zero mean, so the CV is well defined. 
The CV of a random variable $X$ is directly connected to its relative error $r(X) = |X-\expec{X}|/\expec{X}$ as follows. For a given $\epsilon > 0$, using Chebyshev's inequality, we have 
$\prob{r(X) > \epsilon} = \prob{|X-\expec{X}| > \epsilon \expec{X}} \le \frac{\var{X}}{\epsilon^2 \expec{X}^2} = \left(\frac{\cv{X}}{\epsilon}\right)^2$. Smaller the CV, the smaller is the bound for the relative error of a given estimator. Our approach is to choose allocations to different strata so as to optimize a metric based on the CVs of different per-group estimators.


\subsection{Contributions}
\label{sec:contr}
{\bf (1)} We present \cvopt, a novel framework for sampling from data which is applicable to any set of queries that result in multiple answers whose errors need to be simultaneously minimized (a specific example is a group-by query). \cvopt optimizes a cost function based on the weighted aggregate of the CVs of the estimates that are desired by the query.  We present an algorithm (also called \cvopt) that computes a {\em provably optimal} allocation. To our knowledge, this is the first algorithm that results in a provably optimal sample for group-by queries. 
We consider two ways of aggregating CVs -- one based on the $\ell_2$ norm of the CVs, and another based on the $\ell_{\infty}$ norm (the maximum) of the CVs. 

{\bf (2)} \cvopt can be adapted to the following circumstances, in increasing order of generality:

\noindent$\bullet$ With a single aggregation on a single method of grouping attributes (SASG), \cvopt leads to good statistical properties over a range of workloads. The {\em distribution of errors of different groups is concentrated around the mean error}, much more so than prior works~\cite{RL-EDBT2009,AGP00}. As a result, the expected errors of the different per-group estimates are approximately equal, while prior work can lead to some groups being approximated very well while other groups being approximated poorly.

\noindent$\bullet$ With multiple aggregations for the same group-by clause (MASG), and for the general case of multiple aggregations and multiple ways of group-by (MAMG), we provide a way to derive a sample that optimizes for a combined metric on all aggregate queries. A special case is the {\tt Cube} query that is widely used in analytics and decision support workloads.

\noindent$\bullet$ The user is allowed to specify a weight to each answer, allowing her to prioritize different query/group combinations, and use (perhaps uncertain) knowledge about the workload that may be specified by a probability distribution.

{\bf (3)} The samples {\color{red}can be 
  reused to answer various queries that} incorporate
selection predicates that are provided at query time, as well as new
combinations of groupings\footnote{Arbitrary new
  groupings and selection predicates are not supported with a provable
  guarantee. Indeed, we can show no sample that is substantially
  smaller than the data size can provide accurate answers for
  arbitrary queries, since one can reconstruct the original table by
  making such repeated queries.}.


{\bf (4)} We present a detailed experimental study on two
real-world datasets (\openaq and \bikes) that show that \cvopt using
the $\ell_2$ norm of the CVs provides good quality estimates across
all groups in a group-by query, and provides a relative error that is
often up to $5 \times$ smaller than prior
work~\cite{AGP00,sampleandseek,RL-EDBT2009}{\color{red}. Further, 
no prior work can compete with \cvopt as a consistently second best solution for group-by queries.}

\subsection{Prior Work}
\label{sec:related}
The closest work to ours is that of R{\"o}sch and Lehner~\cite{RL-EDBT2009} (which we hence forth call \rl), who propose non-uniform sampling to support group-by queries, where different groups can have vastly different variances. For the case of a single-aggregation and a single group-by, \cvopt is similar to \rl, since \rl also uses the coefficient of variation, in addition to other heuristics. The main difference is that \cvopt provably minimizes the $\ell_2$ norm (or the $\ell_\infty$ norm) of the CVs, whereas \rl is a heuristic without a guarantee on the error, and in fact does not have a well-defined optimization target. The provable guarantees of \cvopt hold even in case of multiple aggregations and/or multiple group-bys, whereas \rl uses a variety of heuristics to handle sample allocation.

Another closely related prior work is ``congressional sampling''~\cite{AGP00} (which we call \cs), which targets sampling for a collection of group-by queries, especially those that consider a ``cube by'' query on a set of group-by attributes. \cs is based on a hybrid of frequency-based allocation (the ``house'') and fixed allocation (the ``senate''). \cvopt differs from \cs through using the coefficients of variation (and hence also the variances) of different strata in deciding allocations, while \cs only uses the frequency. \cvopt results in a provably optimal allocation (even in the case of multiple group-by clauses and aggregates) while \cs does not have such a provable guarantee.

A recent work {\color{blue} ``\sps''}~\cite{sampleandseek} uses a combination of measure-biased sampling and an index to help with low-selectivity predicates. Measure-biased sampling favors rows with larger values along the aggregation attribute. This does not consider the variability within a group in the sampling step -- a group with many rows, each with the same large aggregation value, is still assigned a large number of samples. In contrast, \cvopt favors groups with larger CVs, and emphasizes groups with larger variation in the aggregation attribute. Further, unlike \cvopt, \sps does not provide a sampling strategy that is based on an optimization framework. Our work can potentially be used in conjunction with the index for low-selectivity queries. In our experimental section, we compare with \rl, \cs, and \sps.

All statistics that we use in computing the allocation, including the frequency, mean, and coefficient of variation of the groups, can be computed in a single pass through the data. 
As a result, the overhead of offline sample computation for \cvopt is {\color{red} similar to} that of \cs
and \rl.
We present other related work from the general area of approximate query processing in Section~\ref{sec:other-related}.

%

\smallskip 

\noindent{\bf Roadmap:} We present preliminaries in Section~\ref{sec:prelims}, followed by algorithms for sampling for a single group-by query in Section~\ref{sec:sg} and algorithms for multiple group-by in Section~\ref{sec:mg}. We present an algorithm for different error metrics in Section~\ref{sec:linf}, a detailed experimental study in Section~\ref{sec:experiment}, followed by a survey of other related works, and  
the conclusions.

\section{Preliminaries}
\label{sec:prelims}
For a random variable $X$, let $\expec{X}$ denote its expectation, $\variance{X}$ its variance, $\std{X} = \sqrt{\variance{X}}$ its standard deviation, and $\cv{X} = \frac{\std{X}}{\expec{X}}$ its coefficient of variation. 

The answer to a group-by query can be viewed as a vector of results, one for the aggregate on each group. In case multiple aggregates are desired for each group, the answer can be viewed as a two-dimensional matrix, with one dimension for the groups, and one for the aggregates. In this work, for simplicity of exposition, we focus on the aggregate {\tt AVG}, i.e. the mean. Note that the sample can answer queries involving selection predicates provided at runtime (by simply applying the predicate on the sample) so that it is not possible to precompute the results of all possible queries. Aggregates such as median and variance can be handled using a similar optimization,  and the same sample can be jointly optimized for multiple aggregate functions.

Let the groups of a given query be denoted $1,2,\ldots,r$,~and $\mu_1, \mu_2,\ldots,\mu_r$ denote the mean of value within each group. A sample of a table is a subset of rows of the table. From a sample, we derive estimators $y_i$, of $\mu_i$, for each $i=1\ldots r$. We say that $y_i$ is an unbiased estimator of $\mu_i$ if $\expec{y_i} = \mu_i$. Note that an unbiased estimator does not necessarily mean an estimator that reliably estimates $\mu_i$ with a small error. 

For a group-by query with $r$ groups in data, numbered {\color{blue}$1..r$}, the aggregates can be viewed as an array $\mu = [\mu_1,\mu_2,\ldots,\mu_r]$, the estimates can be viewed as the vector $y = [y_1, y_2, \ldots, y_r]$, and coefficients of variation of the estimates as the vector ${\mathbb C} = \left[\cv{y_1}, \cv{y_2}, \ldots, \cv{y_r}\right]$. We first focus on the overall error equal to the $\ell_2$ norm of the vector ${\mathbb C}$, defined as:
\[
\error(y) = \ell_2({\mathbb C}) =  \sqrt{\sum_{i=1}^r \left(\cv{y_i}\right)^2}
\]

Applying the above metric to the earlier example, we have $\cv{y_1} = 10/1000 = 0.01$, while $\cv{y_2} = 10/100 = 0.1$. This (correctly) evaluates $y_2$ as having a higher contribution to the overall error than $y_1$. If we were to optimize $\ell_2$ norm of the vector $[\cv{y_1}, \cv{y_2}]$, resources would be spent on making $\cv{y_2}$ smaller, at the cost of increasing $\cv{y_1}$. We argue this is the right thing to do, since all results in a group-by query are in some sense, equally important. If we know apriori that some results are more important than others, they can be handled by using a weighting function for results, as we describe further.  We also consider the $\ell_\infty$ norm, defined as: $\mbox{$\ell_\infty({\mathbb C}) =  \max_{i=1}^r \cv{y_i}$}$.


{\color{blue}We can also \textbf{assign weights}} to different results in computing the error. Consider a set of positive real valued numbers, one for each result $i = 1 \ldots r$, $w = \{w_i\}$. The weighted $\ell_2$ metric is:
$\error(y, w) = \ell_2({\mathbb C},w) = \sqrt{\sum_{i=1}^r w_i \cdot \left(\cv{y_i}\right)^2}$,
where a larger $w_i$ indicates a higher accuracy demand for $y_i$.

\remove{

\medskip
\noindent
{\color{red}
{\bf Bounded strata:} We observe that, no master how large the full data is, some strata maybe small, such that the total records is less than the allocated memory. Instead of making assumption about the abundance of the data, we adopt a previous work~\cite{edbt19} to re-allocate the sample space to other strata.
\todo{Will it be better not to introduce a new term without formally defining it in this early section (for reasons: this is not a new contribution of this paper; may cause confusion to readers; increase the paper length.) ? Rather, we give the brief reason, and cite the paper if needed, in the exp section only?}
}

}
\section{Single Group-by}
\label{sec:sg}
\newcommand{\universeA}{\mathcal{A}}
\newcommand{\universeB}{\mathcal{B}}
\newcommand{\universeC}{\mathcal{C}}
\newcommand{\strata}{\mathcal{S}}

\subsection{Single Aggregate, Single Group-By}
\label{sec:sasg}
The first case is when we have a single aggregate query, along with a single group-by clause. Note that grouping does not have to use a single attribute, but could use multiple attributes. For example, {\tt \textbf{SELECT} year,major, \textbf{AVG}(gpa) \textbf{FROM} Student \textbf{GROUP BY} year,major}. {\em Given a budget of sampling $M$ records from a table for a group-by query with $r$ groups, how can one draw a random sample such that the accuracy is maximized?}

We use stratified sampling. In the case of a single group-by clause, {\em stratification directly corresponds to the grouping}. There is a stratum for each group, identified by a distinct value of the combination of group-by attributes. In the above example, there is one stratum for each possible combination of the {\tt (year, major)} tuple. Probabilities of selecting a record from the table are not necessarily equal across different strata, but are equal within a stratum. 

One simple solution, which we call as \senate (used as a component in \cs~\cite{AGP00}), is to split the budget of $M$ records equally among all strata, so that each stratum receives $M/r$ samples. While this is easy to implement and improves upon uniform sampling, this solution has the following drawback. Consider two groups $1$ and $2$ with the same means $\mu_1 = \mu_2$, but with very different standard deviations within the groups, i.e. $\sigma_1 \gg \sigma_2$. Intuitively, it is useful to give more samples to group $1$ than to group 2, to reduce the variance of the estimate within group $1$. However, \senate gives the same number of samples to both, due to which the expected quality of the estimate for group $1$ will be much worse than the estimate for group 2. Intuitively, we need to give more samples to group 1 than group 2, but exactly how much more -- this is answered using our optimization framework.

Before proceeding further, we present a solution to an optimization problem that is repeatedly used in our work. 
\begin{lemma}
\label{lem:opt-helper}
Consider positive valued variables $\mbox{$s_1,s_2,\ldots,s_k$}$ and positive constants $M$, $\alpha_1,\alpha_2,\ldots,\alpha_k$. The solution to the optimization problem: minimize $\sum_{i=1}^k \frac{\alpha_i}{s_i}$ subject to \\
$\mbox{$\sum_{i=1}^k s_i \le M$}$ is given by $s_i = M \cdot \frac{\sqrt{\alpha_i}}{\sum_{j=1}^k \sqrt{\alpha_j}}$.
\end{lemma}
\begin{proof}
Let $s=[s_1,s_2, \ldots, s_k]$. Let $f(s)=\sum_{i=1}^k \frac{\alpha_1}{s_i}$, and $g(s)=(\sum_{i=1}^k s_i) - M$. We want to minimize $f(s)$ subject to the constraint $g(s)=0$. Using Lagrange multipliers: 
\begin{align*}
L(s_1,s_2,...,s_k,\lambda) & = f(s_1,s_2,...,s_k) + \lambda g(s_1,s_2,...,s_k) \\
                           & =\sum_{i=1}^k \frac{\alpha_i}{s_i}+\lambda\left(\sum_{i=1}^{r} s_i - M\right) 
\end{align*}
For each $i=1\ldots r$, we set $\frac{\partial L}{\partial s_i}=0$. Thus  $-\frac{\alpha_i}{s_i^2}+\lambda = 0$, leading to $s_i=\frac{\sqrt{\alpha_i}}{\sqrt{\lambda}}$. By setting $g(s)=0$, we solve for $\lambda$ and get $s_i = M \cdot \frac{\sqrt{\alpha_i}}{\sum_{j=1}^k \sqrt{\alpha_j}}$.
\end{proof}

We now consider how to find the best assignment of sample sizes to the different strata, using an optimization framework. Let $s = [s_1,s_2,\ldots,s_r]$ denote the vector of assignments of sample sizes to different strata. 
\begin{theorem}
\label{thm:sasg}
For a single aggregation and single group-by, given weight vector $w$ and sample size $M$, the optimal assignment of sample sizes is to assign to group $i \in \{1,2,\ldots,r\}$ a sample size $s_i = M \frac{\sqrt{w_i}\sigma_i/\mu_i}{\sum_{j=1}^r \sqrt{w_j}\sigma_j/\mu_j}$ 
\end{theorem}

\begin{proof}
Consider the estimators $\mbox{$y = [y_1,y_2,\ldots,y_r]$}$ computed using the sample. Our objective is the $\ell_2$ error, which requires us to minimize $\mbox{$\sqrt{\sum_{i=1}^r w_i \cdot \left(\cv{y_i}\right)^2}$}$, which is equivalent to minimizing $\sum_{i=1}^r w_i \cdot \left(\cv{y_i}\right)^2$. 

\noindent
The standard deviation of $y_i$ depends on $n_i$, the size of the~$i$th group, $s_i$, the size of the sample assigned to the $i$th group,~and $\sigma_i$, the standard deviation of the values in the $i$th group. By standard results on sampling, e.g. Theorem 2.2 in \cite{Coch77:book}, we have 
\begin{equation*}
\std{y_i} = \sqrt{\frac{\sigma_i^2(n_i-s_i)}{n_i s_i}}, \cv{y_i} = \frac{1}{\mu_i} \sqrt{\frac{\sigma_i^2(n_i-s_i)}{n_i s_i}}
\end{equation*}
Thus, we reduce the problem to minimizing $\sum_{i=1}^r \frac{w_i}{\mu_i^2} \cdot \frac{\sigma_i^2(n_i-s_i)}{n_i s_i}$. Since $r$, $\sigma_i$, $\mu_i$, and $n_i$ are fixed, this is equivalent to minimizing $\sum_{i=1}^r \frac{w_i \sigma_i^2/\mu_i^2}{s_i}$ subject to the condition $\sum_{i=1}^{r} s_i = M$.
Using Lemma~\ref{lem:opt-helper}, and setting $\alpha_i = w_i \sigma_i^2/\mu_i^2$, we see that $s_i$ should be proportional to $\sqrt{w_i} \sigma_i/\mu_i$.
\end{proof}

The above leads to algorithm $\cvoptsasg$ for drawing a random sample from a table $T$, described in Algorithm~\ref{algo:cvopt-sasg}.

\begin{algorithm}[t]
\small
\DontPrintSemicolon

\KwIn{Database Table $T$, group-by attributes $A$, aggregation attribute $d$, weight vector $w$, memory budget $M$.}
\KwOut{Stratified Random Sample $S$}

Let $\universeA$ denote all possibilities of assignments to $A$ that actually occur in $T$. i.e. all strata. Let $r$ denote the size of $\universeA$, and suppose the strata are numbered from $1$ till $r$\\

For each $i = 1\ldots r$, compute the mean and variance of all elements in stratum $i$ along attribute $d$, denoted as $\mu_i, \sigma_i$ respectively. Let $\gamma_i \gets \sqrt{w_i} \sigma_i/\mu_i$\\

$\gamma \gets \sum_{i=1}^r \gamma_i$ \\

\For{$i=1 \ldots r$}
{ $s_i \gets M \cdot \gamma_i/\gamma$ \;
  Let $S_i$ be formed by choosing $s_i$ elements from stratum $i$ uniformly without replacement, using reservoir sampling\;
}

\KwRet{$S = [S_1,S_2,\ldots, S_r]$}
\caption{\cvoptsasg: Algorithm computing a random sample for a single aggregate, single group-by.}
\label{algo:cvopt-sasg}
\end{algorithm}
\subsection{Multiple Aggregates, Single Group-by}
\label{sec:masg}
We next consider the case of multiple aggregations using the same group-by clause.  Without loss of generality, suppose the columns that were aggregated are columns $1,2,\ldots,t$. As before, suppose the groups are numbered $1,2,\ldots,r$. For group $i, 1 \le i \le r$ and aggregation column $j, 1 \le j \le t$, let $\mu_{i,j}, \sigma_{i,j}$ respectively denote the mean and standard deviation of the values in column $j$ within group $i$. Let $n_i$ denote the size of group $i$, and $s_i$ the number of samples drawn from group $i$. Let $y_{i,j}$ denote the estimate of $\mu_{i,j}$ obtained through the sample, and $\cv{y_{i,j}} = \frac{\std{y_{i,j}}}{\mu_{i,j}}$ denote the coefficient of variation of $y_{i,j}$. Further suppose that we are given weights for each combination of group and aggregate, which reflect how important these are to the user \footnote{In the absence of user-input weights, we can assume default weights of 1.}. Let $w_{i,j}$ denote the weight of the combination group $i$ and aggregation column~$j$. Our minimization metric is a weighted combination of the coefficients of variation of all the $r \cdot t$ estimates, one for each group and aggregate combination. 
\[
\error(y,w) = \sqrt{\sum_{j=1}^t \sum_{i=1}^r w_{i,j} \cdot \left(\cv{y_{i,j}} \right)^2}
\]

\begin{theorem}
\label{thm:masg}
Given weights $w$, and total sample size $M$, the optimal assignment of sample sizes $s$ among $r$ groups is to assign to group $i=1,2,\ldots, r$  sample size $s_i = M \frac{\sqrt{\alpha_i}}{\sum_{i=1}^r \sqrt{\alpha_i}}$, where $\alpha_i = \sum_{j=1}^t \frac{w_{i,j} \sigma_{i,j}^2}{\mu_{i,j}^2}$.
\end{theorem}

\main{We omit the proof, that can be found here~\cite{NSPXST-ARXIV2019}. }

\arxiv{
\begin{proof}
We use an approach similar to Theorem~\ref{thm:sasg}.
Let $y = \{y_{i,j}\} 1 \le i \le r, 1 \le j \le t$ denote the matrix of estimators for the means of the multiple aggregates, for different groups. Using the metric of $\ell_2$ error, we have to minimize $\ell_2(\cv{y},w) = \sqrt{\sum_{i=1}^r \sum_{j=1}^t w_{i,j} \cdot \left(\cv{y_{i,j}}\right)^2}$, which is equivalent to minimizing $\sum_{i=1}^r \sum_{j=1}^t w_{i,j} \cdot \left(\frac{\std{y_{i,j}}}{\mu_{i,j}}\right)^2$. 

We note that $\std{y_{i,j}} = \sqrt{\frac{\sigma_{i,j}^2(n_i-s_i)}{n_i s_i}}$, where $\sigma_{i,j}$ is the standard deviation of the $j$th column taken over all elements in group $i$.
Thus, we reduce the problem to minimizing $\sum_{i=1}^r \sum_{j=1}^t w_{i,j} \cdot \frac{\sigma_{i,j}^2(n_i-s_i)}{n_i s_i} \cdot \frac{1}{\mu_{i,j}^2}$. Since $r$, $t$, $\sigma_{i,j}$, $\mu_{i,j}$, and $n_i$ are fixed, this is equivalent to minimizing:

\begin{equation}
\sum_{i=1}^r \sum_{j=1}^t w_{i,j} \frac{\sigma_{i,j}^2}{\mu_{i,j}^2s_i} =  \sum_{i=1}^r \frac{1}{s_i} \sum_{j=1}^t \frac{w_{i,j} \sigma_{i,j}^2}{\mu_{i,j}^2} = \sum_{i=1}^r \frac{\alpha_i}{s_i}
\end{equation}
subject to the condition $\sum_{i=1}^{r} s_i \le M$.
Using Lemma~\ref{lem:opt-helper}, we arrive that the optimal assignment is $s_i = M \frac{\sqrt{\alpha_i}}{\sum_{i=1}^r \sqrt{\alpha_i}}$.
\end{proof}
}

In our formulation, a weight can be assigned to each result in the output, reflecting how important this number is. For instance, if there are $r$ groups and $t$ aggregates desired, then there are $r \times t$ results in the output, and the user can assign a weight for each result, $w_{ij}$ for $i=1 \ldots r$ and $j = 1\ldots t$. A useful special case is when all weights are equal, so that all results are equally important to the user. If the user desires greater accuracy for group 1 when compared to group 2, this can be done by setting weights $w_{1,*}$ to be higher than the weights $w_{2,*}$, say 10 versus 1. The value of weight can also be deduced from a query workload, as we discuss in Section~\ref{sec:workload}.

\section{Multiple Group-Bys}
\label{sec:mg}
\newcommand{\project}{\Pi}
Suppose that we had multiple attribute combinations on which there are group-by clauses. For instance, we may have a query where the student data is being grouped by {\tt major}, and one query where it is being grouped by {\tt year}, and another query where data is grouped by {\tt major} as well as {\tt year}. The additional challenge now is that there are multiple ways to stratify the data, to further apply stratified sampling. For instance, we can draw a stratified sample where data are stratified according to {\tt major} only, or one where data are stratified according to {\tt year}, or one where data are stratified according to both {\tt major} and {\tt year}. Any of these three samples can be used to answer all three queries, but may lead to high errors. For instance, a stratified sample where data is stratified according to year of graduation may lead to poor estimates for a group-by query based on {\tt major}, since it may yield very few tuples or may completely miss some majors. 

Our solution is to pursue a ``finest stratification'' approach where the population is stratified according to the union of all group-by attributes. In the above example, this leads to stratification according to a combination of {\tt major} and {\tt year}, leading to one stratum for each distinct value of the pair {\tt (year,major)}. This will serve group-by queries based solely on {\tt major} or {\tt year}, or a combination of both. The number of samples assigned to each stratum in such a stratification is determined in a principled manner. 
\subsection{Single Aggregate, Multiple Group-By}
\label{sec:sa2g}
We first consider the case of a single aggregate and multiple group-bys, starting with the case of two group-bys and then extend to more than two group-bys. Suppose two queries $Q_1$ and $Q_2$ that aggregate on the same column, using different sets of group-by attributes, $A$ and $B$, respectively. Note that $A$ and $B$ need not be disjoint. For example $A$ can be ({\tt major}, {\tt year}) and $B$ can be ({\tt major}, {\tt zipcode}). If we combined the sets of group-by attributes, we get attribute set $C = A \cup B$. In the above example, $C$ is ({\tt major}, {\tt year}, {\tt zipcode}). Let $\universeA, \universeB, \universeC$ denote the set of all values possible for attributes in $A$, $B$, and $C$ respectively. Note that only those combinations that actually occur within data are considered. 

Our algorithm based on finest stratification stratifies data according to attribute set $C$, leading to a stratum for each~combination of the values of attributes $c \in \universeC$. Samples are~chosen uniformly within a single stratum, but the sampling probabilities in different strata may be different. \emph{Our goal is {\bf not} to get a high-quality estimate for aggregates within each stratum according to $\universeC$. Instead, {\color{blue}it} is to get a high-quality estimate for aggregates for each group in $\universeA$ (query $Q_1$) and~in $\universeB$ (query $Q_2$).} We translate the above goal into an objective function that will help assign sample sizes to each stratum in $\universeC$. 

For each stratum $c \in \universeC$, let $s_c$ denote the number of samples assigned to this stratum, $S_c$ the sample, $\mu_c$ the mean of the aggregation column, and $\sigma_c$ the standard deviation of the aggregation column. Let the sample mean for this stratum be denoted as $y_c = \frac{\sum_{v \in S_c} v}{s_c}$. 
As $C = A \cup B$, $A \in C$. For an assignment $a \in \universeA$ and an assignment $c \in \universeC$, we say $c \in a$ if the attributes in set $A$ have the same values in $a$ and $c$. Let $\universeC(a)$ denote the set of all $c \in \universeC$ such that $c \in a$. For any $c \in \universeC$, let $\project(c,A)$ denote the unique $a \in \universeA$ such that $c \in \universeC(a)$. Similarly, define $\project(c,B)$. 


For query $Q_1$, for group $a \in \universeA$, let $\mu_a$ denote the mean of aggregate column, and $n_a$ denote the size of the group. We desire to estimate $\mu_a$ for each $a \in \universeA$. Suppose the estimate for $\mu_a$ is $y_a$. Similarly, we define $\mu_b$, $n_b$, and $y_b$ for each group $b \in \universeB$. Our objective function is the weighted $\ell_2$ norm of $\{\cv{y_a} | a \in \universeA\} \cup \{\cv{y_b} | b \in \universeB\}$, i.e. 
\[
\sqrt{ \sum_{a \in \universeA} w_a\cdot (\cv{y_a})^2 + \sum_{b \in \universeB} w_b \cdot (\cv{y_b})^2 }
\]

The estimates for each group are derived as $y_a = \frac{\sum_{c \in \universeC(a)} n_c y_c}{\sum_{c \in \universeC(a)} n_c}$, and similarly $y_b = \frac{\sum_{c \in \universeC(b)} n_c y_c}{\sum_{c \in \universeC(b)} n_c}$. Using standard results from stratified sampling~\cite{Coch77:book}, we have:
$\expec{y_a} = \mu_a$, and
\[\variance{y_a} = \frac{1}{n_a^2} \sum_{c \in \universeC(a)} \left[ \frac{n_c^2 \sigma_c^2}{s_c} - n_c\sigma_c^2 \right]\]
\begin{lemma}
\label{lem:samg}
The optimal assignment of sample sizes that minimizes the weighted $\ell_2$ norm of the coefficients of variation of the estimates is: for $d \in \universeC$ the sample size is $s_d = M \cdot \frac{\sqrt{\beta_d}}{\sum_{c \in \universeC} \sqrt{\beta_c}}$, where:
\[\beta_c = n_c^2 \sigma_c^2 \left[\frac{w_{\project(c,A)}}{n_{\project(c,A)}^2 \mu_{\project(c,A)}^2} + \frac{w_{\project(c,B)}}{n_{\project(c,B)}^2 \mu_{\project(c,B)}^2} \right]\]
\end{lemma}

\begin{proof}
Our objective function is the weighted $\ell_2$ norm of the coefficients of variance of all estimators $\{y_a | a \in \universeA\}$ and $\{y_b | b \in \universeB\}$, which we want to minimize over all possibilities of the sample sizes $s = \{s_c | c \in \universeC\}$, subject to $\sum_{c \in \universeC} s_c = M$. Equivalently, we minimize the square of the weighted $\ell_2$ norm:
\begin{align*}
Y(s) & = \sum_{a \in \universeA} w_a \cdot (\cv{y_a})^2 + \sum_{b \in \universeB} w_b \cdot(\cv{y_b})^2
\\
& = \sum_{a \in \universeA} \frac{w_a\cdot \variance{y_a}}{\mu_a^2} + \sum_{b \in \universeB}  \frac{w_b\cdot\variance{y_b}}{\mu_b^2}
\end{align*}
Using the exp. and variance of $y_a$, we can rewrite $Y(s)$: 
\begin{align*}
Y(s) & = \sum_{a \in \universeA} \frac{w_a\cdot\sum_{c \in \universeC(a)} \left[\frac{n_c^2 \sigma_c^2}{s_c}  - n_c\sigma_c^2 \right]}{n_a^2 \mu_a^2} \\
    & + \sum_{b \in \universeB}  \frac{w_b\cdot \sum_{c \in \universeC(b)} \left[ \frac{n_c^2 \sigma_c^2}{s_c} - n_c\sigma_c^2\right]}{n_b^2 \mu_b^2} 
\end{align*}
This is equivalent to minimizing:
\begin{align*}
Y'(s) & = \sum_{a \in \universeA} \sum_{c \in \universeC(a)} \frac{w_a  n_c^2 \sigma_c^2}{s_c n_a^2 \mu_a^2} + \sum_{b \in \universeB} \sum_{c \in \universeC(b)} \frac{w_b  n_c^2 \sigma_c^2}{s_c n_b^2 \mu_b^2}
\\
& = \sum_{c \in \universeC} \left[ \frac{w_{\project(c,A)} n_c^2 \sigma_c^2}{s_c n_{\project(c,A)}^2 \mu_{\project(c,A)}^2} + \frac{w_{\project(c,B)} n_c^2 \sigma_c^2}{s_c n_{\project(c,B)}^2 \mu_{\project(c,B)}^2} \right]
\end{align*}
We note that the problem turns to: minimize $Y'(s) = \sum_{c \in \universeC} \frac{\beta_c}{s_c}$ subject to $\sum_{c \in \universeC} s_c = M$. Using Lemma~\ref{lem:opt-helper}, we arrive that the optimal assignment of sample size is $s_c = M \cdot \frac{\sqrt{\beta_c}}{\sum_{d \in \universeC} \sqrt{\beta_d}}$.
\end{proof}

{\bf An example:} Consider a query $Q_1$ that groups by {\tt major} and aggregates by {\tt gpa}, and another query $Q_2$ that groups by {\tt year}, and aggregates by {\tt gpa}. Suppose each group in each query has the same weight $1$. The above algorithm stratifies according to the {\tt (major,year)} combination. For a stratum where {\tt major} equals $m$ and {\tt year} equals $y$, sample size $s_{m,y}$ is proportional to:
\[\beta_{m,y} = n_{m,y}^2 \sigma_{m,y}^2 \left[\frac{1}{n_{m,*}^2 \mu_{m,*}^2} + \frac{1}{n_{*,y}^2 \mu_{*,y}^2} \right]\]
Where $n_{m,y}, n_{m,*}, n_{*,y}$ are respectively the number of elements with {\tt major}~$=m$  and {\tt year}~$=y$, the number of elements with {\tt major}~$=m$, and the number of elements with {\tt year}~$=y$, respectively. Similarly for $\mu_{m,y}, \mu_{m,*}, \mu_{*,y}$.

{\bf Example 2:} Consider a query $R_1$ that groups by {\tt major, year} and aggregates by {\tt gpa}, and another query $R_2$ that groups by {\tt zipcode, year}, and aggregates by {\tt gpa}. Suppose all groups in both queries share the same weight~$1$. The above algorithm asks to stratify according to {\tt (major, zipcode, year)} combination. For a stratum where {\tt major} equals $m$, {\tt zipcode} equals $z$ and {\tt year} equals $y$, sample size $s_{m,z,y}$ is proportional to:
\[\beta_{m,z,y} =  n_{m,z,y}^2 \sigma_{m,z,y}^2 \left[\frac{1}{n_{m,*,y}^2 \mu_{m,*,y}^2} + \frac{1}{n_{*,z,y}^2 \mu_{*,z,y}^2} \right]\]
Where $n_{m,z,y}$, $n_{m,*,y}$, and $n_{*,z,y}$ are respectively the number of elements with {\tt major} equal to $m$ and {\tt zipcode} equal to $z$ and {\tt year} equal to $y$, the number of elements with {\tt major} equal to $m$ and {\tt year} equal to $y$, and the number of elements with {\tt zipcode} equal to $z$ and {\tt year} equal to $y$, respectively. Similarly for $\mu_{m,z,y}, \mu_{m,*,y}, \mu_{*,z,y}$.


\paragraph*{Generalizing to Multiple Group-Bys} 
\label{sec:1amg}
Suppose there were multiple group-by queries with attribute sets $A_1,A_2, \ldots, A_k$. The algorithm stratifies according to attribute set $C = \bigcup_{i=1}^k A_i$. For $i=1\ldots k$, let $\universeA_i$ denote the universe of all assignments to attributes in $A_i$ and $\universeC$ the universe of all possible assignments to attributes in $C$. Note that only those assignments that exist in the data need be considered. Extending the above analysis for the case of two group-bys, we get that the optimal assignment of samples as follows. For each $c \in \universeC$, stratum $c$ is assigned sample size proportional to the square root of:
\[\beta_c = n_c^2 \sigma_c^2 \sum_{i=1}^k \frac{w_{\project(c,A_i)}}{n_{\project(c,A_i)}^2 \mu_{\project(c,A_i)}^2}\]
\main{The proofs are similar to the case of two group-by queries, and details are available here~\cite{NSPXST-ARXIV2019}.}
\arxiv{
The proof of above result is similar to the case of Lemma~\ref{lem:samg}. We minimize the $\ell_2$ norm:
\begin{align*}
Y(s) & = \sum_{a_i \in \universeA_i} \frac{w_{a_i}\cdot\sum_{c \in \universeC(a_i)} \left[\frac{n_c^2 \sigma_c^2}{s_c}  - n_c\sigma_c^2 \right]}{n_{a_i}^2 \mu_{a_i}^2}
\end{align*}
This is equivalent to minimizing:
\begin{align*}
Y'(s) & = \sum_{a_i \in \universeA_i} \sum_{c \in \universeC(a_i)} \frac{w_{a_i}  n_c^2 \sigma_c^2}{s_c n_{a_i}^2 \mu_{a_i}^2}
\\
& = \sum_{c \in \universeC}  \sum_{i=1}^k  \frac{w_{\project(c,A_i)} n_c^2 \sigma_c^2}{s_c n_{\project(c,A_i)}^2 \mu_{\project(c,A_i)}^2}
= \sum_{c \in \universeC} \frac{\beta_c}{s_c}
\end{align*}
Using Lemma~\ref{lem:opt-helper}, we have the result proved.

}

\paragraph*{Cube-By Queries} An important special case of multiple group-by queries, often used in data warehousing, is the {\em cube-by} query. The cube-by query takes as input a set of attributes $A$ and computes group-by aggregations based on the entire set $A$ as well as every subset of $A$. For instance, if $A$ was the set {\tt major, year} and the aggregation column is $A$, then the cube-by query poses four queries, one grouped by {\tt major} and {\tt year}, one grouped by only {\tt major}, one grouped by only {\tt year}, and the other without a group-by (i.e. a full table query). Our algorithm for multiple group-by can easily handle the case of a cube-by query and produce an allocation that optimizes the $\ell_2$ norm of the CVs of all estimates. We present an experimental study of cube-by queries in Section~\ref{sec:experiment}.

\subsection{Multiple Aggregates, Multiple Group-Bys}
\label{sec:2a2g}
Suppose two queries, $Q_1$, $Q_2$, that aggregate on the different columns $d_1$ and $d_2$ and also use different sets of group-by attributes, $A$ and $B$ that may be overlapping. \eg $Q_1$ can aggregate {\tt gpa} grouped by  ({\tt major}, {\tt year}) and $Q_2$ can aggregate {\tt credits} grouped by ({\tt major}, {\tt zipcode}). 

We stratify the data according to attribute set $C = A \cup B$. As in Section~\ref{sec:sa2g}, let $\universeA, \universeB, \universeC$ denote the set of all values possible for attributes in $A$, $B$, and $C$ respectively. Also, for $a \in \universeA$, $b \in \universeB$, $c \in \universeC$, let $\universeC(a)$, $\universeC(b)$, $\project(c,A)$ and $\project(c,B)$ be defined as in Section~\ref{sec:sa2g}.

For each $c \in \universeC$, let $n_c$ denote the number of data elements in this stratum, $\sigma_{c,1}$ the variance of the $d_1$ column among all elements in this stratum, and $\sigma_{c,2}$ the variance of the $d_2$ column in this stratum. Let $s_c$ denote the number of samples assigned to this stratum, and $y_{c,1}$ and $y_{c,2}$ respectively denote the sample means of the columns $d_1$ and $d_2$ among all elements in stratum $c$ respectively.

For each $a \in \universeA$, we seek to estimate $\mu_{a,1}$, the mean of the $d_1$ column among all elements in this group. The estimate, which we denote by $y_{a,1}$ is computed as $\frac{\sum_{c \in \universeC(a)} n_c y_{c,1}}{\sum_{c \in \universeC(a)} n_c}$. Similarly for each $b \in \universeB$, we seek to estimate $\mu_{b,2}$ the mean of the $d_2$ column among all elements in this group. Let $y_{b,2}$ be this estimate. Our optimization metric is the weighted $\ell_2$ norm of the coefficients of variation of all estimates:
\[\mathcal{L}=\sum_{a \in \universeA} w_{a,1} \cdot (\cv{y_{a,1}})^2 + \sum_{b \in \universeB} w_{b,2} \cdot (\cv{y_{b,2}})^2\]

\begin{lemma}
\label{lem:mamg}
For two group-by and two aggregates, the optimal assignment of sample sizes that minimizes the weighted $\ell_2$ norm of the coefficients of variation of the estimates is: for $d \in \universeC$ the sample size is $s_d = M \cdot \frac{\sqrt{\beta_d}}{\sum_{c \in \universeC} \sqrt{\beta_c}}$, where 
\[\beta_c = n_c^2 \left[\frac{w_{\project(c,A),1}\sigma_{c,1}^2 }{n_{\project(c,A)}^2 \mu_{\project(c,A),1}^2} + \frac{w_{\project(c,B),2}\sigma_{c,2}^2}{n_{\project(c,B)}^2 \mu_{\project(c,B),2}^2} \right]\]
\end{lemma}

\main{The proof can be found here~\cite{NSPXST-ARXIV2019}.}

\arxiv{
\begin{proof}
	Our objective function is the weighted $\ell_2$ norm of the coefficients of variance of all estimators $\{y_{a,1} | a \in \universeA\}$ and $\{y_{b, 2} | b \in \universeB\}$, which we want to minimize over all possibilities of the vector of sample sizes $s = \{s_c | c \in \universeC\}$, subject to $\sum_{c \in \universeC} s_c = M$. Equivalently, we minimize the square of the weighted $\ell_2$ norm:
	\begin{align*}
	Y(s) & = \sum_{a \in \universeA} w_{a,1} \cdot (\cv{y_{a,1}})^2 + \sum_{b \in \universeB} w_{b,2} \cdot(\cv{y_{b,2}})^2 \\
	& = \sum_{a \in \universeA} \frac{w_{a,1} \cdot \variance{y_{a,1}}}{\mu_{a,1}^2} + \sum_{b \in \universeB}  \frac{w_{b,2}\cdot\variance{y_{b,2}}}{\mu_{b,2}^2}
	\end{align*}
	Using the expected value and variance of $y_a$, we can rewrite $Y(s)$: 
	
\[Y(s) = \sum\limits_{a \in {\rm{ \universeA}}} {\frac{{{w_{a,1}} \cdot \sum\limits_{c \in {\rm{ \universeC}}(a)} {\left[ {\frac{{n_c^2\sigma _{c,1}^2}}{{{s_c}}} - {n_c}\sigma _{c,1}^2} \right]} }}{{n_a^2\mu _{a,1}^2}}}  + \sum\limits_{b \in {\rm{ \universeB}}} {\frac{{{w_{b,2}} \cdot \sum\limits_{c \in {\rm{ \universeC}}(b)} {\left[ {\frac{{n_c^2\sigma _{c,2}^2}}{{{s_c}}} - {n_c}\sigma _{c,2}^2} \right]} }}{{n_b^2\mu _{b,2}^2}}} \]
	This is equivalent to minimizing:
	\begin{align*}
	Y'(s) & = \sum\limits_{a \in {\rm{ \universeA}}} {\sum\limits_{c \in {\rm{ \universeC}}(a)} {\frac{{{w_{a,1}}n_c^2\sigma _{c,1}^2}}{{{s_c}n_a^2\mu _{a,1}^2}}} }  + \sum\limits_{b \in {\rm{ \universeB}}} {\sum\limits_{c \in {\rm{ \universeC}}(b)} {\frac{{{w_b}n_c^2\sigma _{c,2}^2}}{{{s_c}n_b^2\mu _{b,2}^2}}} }
	\\
	& = \sum\limits_{c \in {\rm{ \universeC}}} {\left[ {\frac{{{w_{{\rm{\project}}(c,A)}}n_c^2\sigma _{c,1}^2}}{{{s_c}n_{{\rm{\project}}(c,A)}^2\mu _{{\rm{\project}}(c,A),1}^2}} + \frac{{{w_{{\rm{\project}}(c,B)}}n_c^2\sigma _{c,2}^2}}{{{s_c}n_{{\rm{\project}}(c,B)}^2\mu _{{\rm{\project}}(c,B),2}^2}}} \right]}
	\\
	& = \sum_{c \in \universeC} \frac{\beta_c}{s_c}
	\end{align*}
	Subject to $\sum_{c \in \universeC} s_c = M$. Using Lemma~\ref{lem:opt-helper}, we arrive that the optimal assignment of sample size is $s_c = M \cdot \frac{\sqrt{\beta_c}}{\sum_{d \in \universeC} \sqrt{\beta_d}}$.
\end{proof}
}

We can generalize this to the case of more than two aggregations, and/or more than two group-bys. Suppose there were $k$ group-by queries $Q_1, Q_2, \ldots, Q_k$, with attribute sets $A_1,A_2, \ldots, A_k$. Each query $Q_i$ has multiple aggregates on a set of columns denoted as $\mathcal{L}_i$. In this case, the algorithm stratifies according to attribute set $C = \bigcup_{i=1}^k A_i$. For $i=1\ldots k$, let $\universeA_i$ denote the universe of all assignments to attributes in $A_i$ and $\universeC$ the universe of all possible assignments to attributes in $C$. Note that only those assignments that exist in the data need be considered.  Extending the analysis from Section~\ref{sec:masg},~\ref{sec:1amg}, and~\ref{sec:2a2g}, we get that the optimal assignment of samples is as follows.  For each $c \in \universeC$, stratum $c$ is assigned sample size proportional to the square root of 
$$
\beta_c = n_c^2 \sum_{i=1}^k \left(\frac{1}{n_{\project(c,A_i)}^2 }\sum_{\ell\in \mathcal{L}_i}\frac{w_{\project(c,A_i),\ell}\cdot \sigma_{c,\ell}^2 }{\mu_{\project(c,A_i),\ell}^2}\right).
$$


\subsection{Using A Query Workload}
\label{sec:workload}
How can one use (partial) knowledge of a query workload to improve
sampling? A query workload is a probability distribution of expected
queries, and can be either collected from historical logs or created
by users based on their experience. In the presence of a query workload, 
we show how to construct a sample that is optimized for this workload.
We focus on the case of multiple aggregations, multiple group-by
(MAMG), as others such as SASG and MASG are special
cases. Our approach is to use the query workload 
to deduce the weights that we use in the weighted optimization for group-by queries.

\begin{table}[t]
\caption{An example \texttt{Student} table}
\label{tab:student}
\centering
\begin{tabular}{|c|c|c|l|l|l|l|}
\hline
{\bf id} & {\bf age} & {\bf GPA} & {\bf SAT} & {\bf major} & {\bf college}\\
\hline  
1 & 25 & 3.4 & 1250 & CS & Science\\
\hline
2 & 22 & 3.1 & 1280 & CS & Science\\
\hline
3 & 24 & 3.8 & 1230 & Math & Science\\
\hline
4 & 28 & 3.6 & 1270 & Math & Science\\
\hline
5 & 21 & 3.5 & 1210 & EE & Engineering\\
\hline
6 & 23 & 3.2 & 1260 & EE & Engineering\\
\hline
7 & 27 & 3.7 & 1220 & ME &Engineering\\
\hline
8 & 26 & 3.3 & 1230 & ME &Engineering\\
\hline
\end{tabular}
\end{table}

\begin{table}[t]
\caption{An example query workload on the \texttt{Student} table}
\label{tab:workload}
\centering
\begin{tabular}{|l|c|}
\hline
{\bf \hspace{20mm}Queries} & {\bf repeats}\\
\hline
A: \texttt{\textbf{SELECT} AVG(age), AVG(gpa)} & \\
\ \ \ \ \texttt{\textbf{FROM} Student \textbf{GROUP BY} major} &20\\
\hline
B:  \texttt{\textbf{SELECT} AVG(age), AVG(sat)} & \\
\ \ \ \ \texttt{\textbf{FROM} Student \textbf{GROUP BY} college} & 10\\
\hline
C: \texttt{\textbf{SELECT} AVG(gpa) \textbf{FROM} Student} & \\
\ \ \ \ \texttt{\textbf{GROUP BY} major \textbf{WHERE} college=Science} & 15\\
\hline
\end{tabular}
\end{table}

\begin{table}[t]
\caption{Aggregation groups and their frequencies
    produced from the example workload (Table~\ref{tab:workload})}
\label{tab:aggregation-groups}
\centering
{\small
\begin{tabular}{|c|c|}
\hline 
 \textbf{Aggregation groups} & \textbf{Frequency} \\
\hline
  \begin{tabular}[c]{@{}c@{}c@{}}
 	{(age, major=CS), (age, major=Math), (age, major=EE)} \\
    {(age, major=ME), (GPA, major=EE ), (GPA, major=ME)} \\
  \end{tabular} & 25 \\
\hline
(GPA, major=CS), (GPA, major=Math) & $35$ \\
\hline
  \begin{tabular}[c]{@{}c@{}}
    {(age, college=Science), (age, college=Engineering)} \\
    {(SAT, college=Science), (SAT, college=Engineering)} \\
  \end{tabular}  & $10$ \\
  \hline
\end{tabular}
}
\end{table}


{\bf An Example:}
Consider an example \texttt{Student} table and its query workload
shown in Tables~\ref{tab:student} and~\ref{tab:workload}.
The workload has $45$
group-by queries, of which three are distinct, named A, B, and C.
Each group-by query stratifies its aggregation
columns into a collection of mutually disjoint \emph{aggregation groups}.
Each aggregation group is identified by a tuple of $(a,b)$, where $a$
is the aggregation column name and $b$ is one value assignment to the
group-by attributes. For example,  Query A stratifies the \texttt{age}
column into four aggregation groups: \emph{(age,major=CS)}, \emph{(age,
  major=Math)}, \emph{(age,major=EE)}, and \emph{(age,major=ME)}. Each
aggregation group is a subset of elements in the aggregation column, \eg aggregation group \emph{(age,major=CS)} is the set
$\{25,22\}$. 
One aggregation group may appear more than once,
because one query may occur multiple times
and different queries may share aggregation
group(s). Our preprocessing is to deduce all
the aggregation groups and 
their frequencies from the workload. Table~\ref{tab:aggregation-groups} shows the result
of the example
workload. We then use each aggregation
group's frequency as its weight in the
optimization framework for \cvopt sampling.\looseness=-1

	\section{\cvoptinf and Extensions}

\label{sec:linf}
We now consider {\color{blue} minimizing} for the maximum of the CVs, or the $\ell_\infty$ norm of all $\cv{y_i}$ for different groups $i$. 
That is:
$$
\ell_\infty({\mathbb C}) =  \max_{i=1}^r \cv{y_i} = \max_{i=1}^r
\frac{\sigma_i}{\mu_i} \sqrt{ \frac{n_i-s_i}{n_i s_i}}
$$
subject to the constraint $\sum_{i=1}^{r} s_i \leq M$.
One obvious benefit of using $\ell_\infty({\mathbb C})$ as the objective
function is that the relative errors of different groups are expected to be the same.

The above problem has integrality constraints and is hard to optimize exactly.
In the rest of this section, we present an
efficient algorithm that relaxes the integrality constraints and 
assumes that $s_i$s can have real values. Note that we assume every $\sigma_i > 0$, 
any group $i$ where $\sigma_i=0$ can be treated as a special case, since all its values are equal,
and there is no need to maintain a sample of that group. 

\textbf{An efficient algorithm:}
Consider a collection of $r$ functions 
$$f_i(x) = \frac{\sigma_i}{\mu_i} \sqrt{ \frac{n_i-x}{n_i x}}, \ \ i=1,2,\ldots,r$$
where $x\in [0,M]$ is a real number. Our goal is find an assignment of
$x=x_i$ for each $f_i$, that minimizes 
$$
\max_{i=1}^r f_i(x_i)= \max_{i=1}^r
\frac{\sigma_i}{\mu_i} \sqrt{ \frac{n_i-x_i}{n_i x_i}}
$$
and $\sum_{i=1}^{r}x_i \leq M$.
Let $x_1^{\ast}, x_2^{\ast}, \ldots, x_r^{\ast}$ denote
such assignments that solve this continous optimization problem. 

\begin{lemma}
\label{lem:l-inf-cv}
$f_1(x_1^{\ast}) = f_2(x_2^{\ast}) = \ldots = f_r(x_r^{\ast})$
\end{lemma}

\main{The lemma's proof by contradiction can be found here~\cite{NSPXST-ARXIV2019}.}
\arxiv{
\begin{proof}
We use proof by contradiction. Suppose the claim in the lemma is not
true, then without losing generality, say $f_1(x_1^{\ast}),
f_k(x_k^{\ast}), \ldots, f_k(x_k^{\ast})$, for some $k< r$, are all the 
largest and $f_r(x_r^{\ast})$ is the smallest. 
Then, because each $f_i$ is a strictly decreasing function, 
we can always reduce the value of each $x_i^{\ast}$, $i\leq k$, for
some amount $c$ and increase the value of  $x_r^{\ast}$
by amount $kc$, such that (1) every $f_i$ is decreased, $i\leq k$.
(2) $f_r$ is increased, and  (3) $\max_{i=1}^r f_i$ is
decreased. This is a contradiction, since we already know 
$\max_{i=1}^r f_i(x_i^{\ast})$ was minimized. 
\end{proof}
}
%
%
%
Following Lemma~\ref{lem:l-inf-cv}, we can have 
$$
\frac{\sigma_1}{\mu_1} \sqrt{ \frac{n_1-x_1^\ast}{n_1 x_1^\ast}} = 
\frac{\sigma_2}{\mu_2} \sqrt{ \frac{n_2-x_2^\ast}{n_2 x_2^\ast}} = 
\ldots
\frac{\sigma_r}{\mu_r} \sqrt{ \frac{n_r-x_r^\ast}{n_r x_r^\ast}}
$$
\begin{eqnarray}
\label{eqn:continous}
\Longrightarrow\frac{d_1}{x^\ast_1/\bar{x}^\ast_1}
=\frac{d_2}{x^\ast_2/\bar{x}^\ast_2}
=\ldots
=\frac{d_r}{x^\ast_r/\bar{x}^\ast_r}
\end{eqnarray}
where $d_i = \frac{(\sigma_i/\mu_i)^2}{n_i}$ and $\bar{x}^\ast_i = n_i - x^\ast_i$. Let $D = \sum_{i=1}^r d_i$. 
{\color{blue}By} Equations~\ref{eqn:continous}, each $x^\ast_i/\bar{x}^\ast_i$ is 
proportional
to $d_i$, i.e, 
${x^\ast_i}/{\bar{x}^\ast_i} = q^\ast\cdot {d_i}/{D}$ for some  real number
constant
$q^\ast \in [0,n]$. Namely,
$$ x^\ast_i = \frac{q^\ast \cdot d_i/D}{1+q^\ast\cdot d_i/D}n_i. $$
Our approach to minimize $\ell_\infty({\mathbb C})$ is to perform
a binary search for the largest integer $q\in [0,n]$ that approximates
$q^\ast$, {\color{blue} in such}
$$
\sum_{i=1}^r x_i 
=
\sum_{i=1}^r \frac{q\cdot d_i/D}{1+q\cdot d_i/D}n_i
\leq M.
$$ 
If the binary search returns $q=0$, we set $q=1$. 
We then assign each $s_i = \left\lceil\frac{x_i}{\sum_{j=1}^r x_j}M\right\rceil$.
Clearly, the total time cost for finding the set of $s_i$ values 
is $O(r \log n)$: There are a total of $O(\log n)$ binary
search steps where each step takes $O(r)$ time.

\textbf{Extension to Other Aggregates.}
\label{sec:ext}
Thus far, our discussion of the \cvopt framework has focused on group-by
queries using the AVG aggregates (COUNT and SUM are very similar). 
\cvopt can be extended to other aggregates as well. 
To use the framework for an aggregate, we need to: (1)~have the per-group CV of the aggregate of interest well defined,
and (2)~ensure that it is possible to compute the CV of a stratum using statistics stored for strata in finer stratification of this stratum.
Hence, the method can potentially be extended to aggregates such as per-group median and variance.

\section{Experimental Evaluation}
\label{sec:experiment}

{\color{red} We evaluate \cvopt for approximate query
  processing using Hive~\cite{hive} as the underlying data
  warehouse. The first phase is an offline sample computation that
  performs two passes of the data. The first pass computes some
  statistics for each group, and the second pass uses these
  statistics as input for \cvopt to decide the sample sizes for different
  groups and to draw the actual sample. 
  
  
%
%
  In the second phase, the sample obtained using \cvopt is used to (approximately) answer queries.
  The sample from \cvopt is representative and can answer queries that involve selection predicates provided at query time, as
  well as new combinations of groupings, without a need for recomputing the sample at query time.
  Overall, we can expect the overhead of the offline sampling phase to be small when compared with
  the time saved through sample-based query processing.
}

\remove{
  incorporate selection predicates provided at query time, as
  well as new combinations of groupings, it can answer
  queries with good accuracy and the overhead by the offline
  phase thus becomes negligible compared to the time we save 
  from sample-based query processing.\looseness=-1
%
}

We collected two real-world datasets, \openaq and \bikes.
\openaq~\cite{OpenAQ} is a collection of the air quality measurements of different substances, such as carbon monoxide, sulphur dioxide, \etc The data consists of about 200 million records, collected daily from ten thousand locations in 67 countries from 2015 to 2018. 
\bikes is data logs from Divvy~\cite{bikeshare}, a Chicago bike share company. Customers can pick up a bike from one Divvy station kiosk and return it to any station at their convenience. The dataset contains information about these bike rides and also has some user information, such as gender or birthday. We analyzed all subscribers' data from 2016 to 2018, for a total of approximately 11.5 million records. The datasets are stored in the database as 2 tables {\tt OpenAQ} and {\tt Bike}. Throughout this section, we introduce queries to those 2 tables, annotated with "AQ" and "B" prefixes, respectively.


Approximate answers from sample table are compared to the ground truth, derived using an exact computation from the full data. Let $x$ and $\bar{x}$ be the ground-truth and approximate answer, respectively. We use the relative error $\vert \bar{x}-x \vert /x$ as the error metric for each group. 

We implemented algorithms \uniform, \sps, \cs and \rl to compare with \cvopt (see Section~\ref{sec:related} for a brief overview).
\uniform is the unbiased sampler, which samples records uniformly without replacement from the base table.
\rl is the algorithm due to Rosch and Lehner~\cite{RL-EDBT2009}.
\cs is congressional sampling algorithm due to~\cite{AGP00}.
\sps is from~\cite{sampleandseek}, after applying appropriate normalization to get an unbiased answer. 
\cvopt is the implementation of our $\ell_2$ optimal sampler. We also report results from \cvoptinf, the $\ell_{\infty}$-optimal sampler. Unless otherwise specified, each method draws a $1\%$ sample. 
{\color{red}
Each reported result is the average of 5 identical and independent repetitions of an experiment.
}

\subsection{Accuracy of Approximate Query Processing}
\label{exp:accuracy}
\begin{query}[b]
\renewcommand{\ttdefault}{pcr}
\begin{lstlisting}
WITH bc18 AS (
  SELECT country, AVG(value) AS avg_value, 
         COUNT_IF(value > 0.04) AS high_cnt
  FROM   openaq	WHERE  parameter = 'bc' 
    AND  YEAR(local_time) = 2018 
  GROUP BY country ), 
bc17 AS (
  SELECT country, AVG(value) AS avg_value, 
         COUNT_IF(value > 0.04) AS high_cnt
  FROM   openaq	WHERE  parameter = 'bc' 
    AND  YEAR(local_time) = 2017
  GROUP BY country )
SELECT country,
  bc18.avg_value - bc17.avg_value AS avg_incre,
  bc18.high_cnt - bc17.high_cnt AS cnt_incre
FROM bc18 JOIN bc17
ON bc18.country = bc17.country
\end{lstlisting}
\caption{Changing of \textit{bc} overtime for each country.}
\label{query:openaq_masg_complex}
\end{query}
\begin{query}[b]
\renewcommand{\ttdefault}{pcr}
\begin{lstlisting}
SELECT country, parameter, unit,
       SUM(value) agg1, COUNT(*) agg2
FROM OpenAQ
GROUP BY country, parameter, unit
\end{lstlisting}
\caption{MASG query to \openaq table.}
\label{query:openaq_masg}
\end{query}

\begin{query2}[b]
\renewcommand{\ttdefault}{pcr}
\begin{lstlisting}
SELECT from_station_id, 
       AVG(age) agg1, AVG(trip_duration) agg2
FROM Bikes	WHERE age > 0
GROUP BY from_station_id
\end{lstlisting}
\caption{MASG query to \bikes table.}
\label{query:trips_masg}
\end{query2}

\begin{figure}[t]
\centering
\includegraphics[width=\graphwidth]{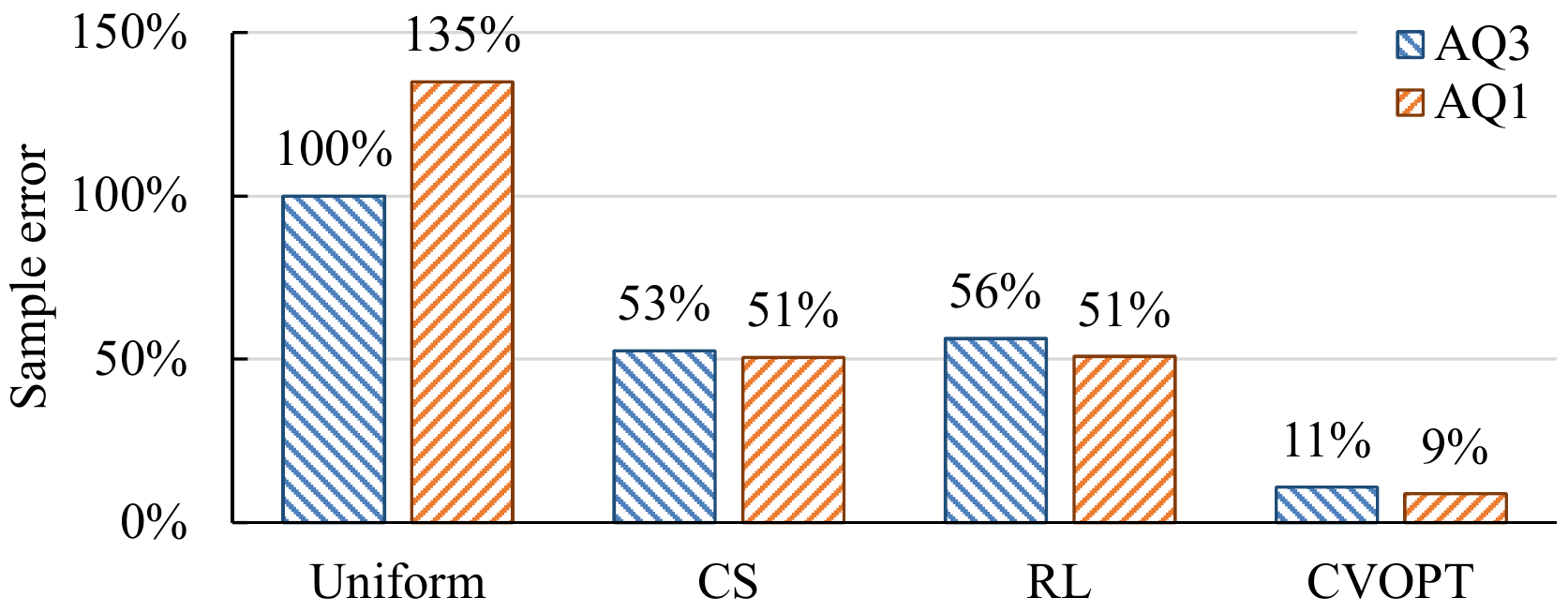}
\caption{Maximum error for \masg query AQ\ref{query:openaq_masg_complex} and \sasg query AQ\ref{query:openaq_sasg} using a $1\%$ sample.}
\label{fig:openaq_max}
\end{figure}

The quality of a sample is measured by the accuracy of the approximate answers using the sample. 
{\color{red}
	We introduce {\masg} queries AQ\ref{query:openaq_masg_complex}, AQ\ref{query:openaq_masg} and B\ref{query:trips_masg}.}
AQ\ref{query:openaq_masg_complex} is a relatively complex, realistic query computing the changes of both the average and the number of days with high level of black carbon \textit{(bc)}, for each country between 2017 and 2018. The query contains
different aggregates, multiple table scans, and a join operator. 
\remove{Note that our sample is materialized, it serves multiple table scans, as
well as different queries.}
AQ\ref{query:openaq_masg} and B\ref{query:trips_masg} are simpler examples of \masg query, which has multiple aggregate functions sharing the group-by.

Figure~\ref{fig:openaq_max} shows the maximum
errors of the approximate answers of
query AQ\ref{query:openaq_masg_complex} using a $1\%$ sample.
We report the maximum error across all answers. \cvopt shows a
significant improvement over other methods. It has a maximum error
of $8.8\%$ while \cs and \rl have error of as much as
$50\%$. With the same space budget, the error of \uniform can be as
large as $135\%$, as some groups are poorly represented.
Similar improvements are observed with other \masg queries. 
For AQ\ref{query:openaq_masg}, the maximum errors of \cs, \rl and \cvopt are $10.1\%$, $29.5\%$ and $5.9\%$ respectively. 
For B\ref{query:trips_masg} the maximum errors of \cs, \rl and \cvopt are $11.7\%$, $8.8\%$ and $7.7\%$, respectively.\looseness=-1


\begin{query}[b]
\renewcommand{\ttdefault}{pcr}
\begin{lstlisting}
SELECT country, parameter, unit, AVG(value)
FROM OpenAQ
WHERE HOUR(local_time) BETWEEN 0 AND 24
GROUP BY country, parameter, unit
\end{lstlisting}
\caption{Average value.}
\label{query:openaq_sasg}
\end{query}

\begin{query2}[b]
\renewcommand{\ttdefault}{pcr}
\begin{lstlisting}
SELECT from_station_id, AVG(trip_duration)
FROM Bikes	WHERE trip_duration > 0
GROUP BY from_station_id
\end{lstlisting}
\caption{Average trip duration from each station.}
\label{query:trips_sasg}
\end{query2}

\begin{query}[b]
\renewcommand{\ttdefault}{pcr}
\begin{lstlisting}
SELECT AVG(value),  
       country, 
       CONCAT(month, '_', year) 
FROM (SELECT value,
             MONTH(local_time) AS month, 
             YEAR(local_time) AS year, 
             country 
      FROM  OpenAQ	WHERE parameter = 'co' ) 
GROUP BY country, month, year
\end{lstlisting}
\caption{Average carbon monoxide.}
\label{query:openaq_sasg_complex}
\end{query}


{\color{red} We present queries AQ\ref{query:openaq_sasg},
  B\ref{query:trips_sasg} and AQ\ref{query:openaq_sasg_complex} as
  case-studies for \sasg query.} AQ\ref{query:openaq_sasg_complex}
is a realistic example, while AQ\ref{query:openaq_sasg} and B\ref{query:trips_sasg} are simpler
examples.  Figure~\ref{fig:openaq_max} shows
the maximum errors for AQ\ref{query:openaq_sasg} using a $1\%$ sample.
For both \sasg and \masg queries, \cvopt yields the lowest
error.  While \cvopt has 11\% sample error, \cs
and \rl have large errors of more than 50\%.  
{\color{red} For \sasg query, \eg
  AQ\ref{query:openaq_sasg_complex}, \cvopt and \rl share similar
  objective functions. However, \rl assumes that the size of a group is always large,
   and in allocation sample sizes, does not take the group size 
   into account (it only uses the CV of elements in the group).
   However real data, including the \openaq dataset, may contain small groups, where 
   \rl may allocate a sample size greater than the group size. \cvopt does not make such an
   assumption, and hence often leads to better quality samples than \rl, even for the case of \sasg.}
   
   \remove{
  size for each group is always no larger than the group size, which,
  however, may not be true in realworld data sets (e.g., the \openaq
  data). Thus, they waste the
  overly allocated
  memory space. Make no such assumption, \cvopt 
  always reuses any overly allocated memory space for groups that have
  enough records to be drawn into the sample.}
\uniform has largest error of 100\%,
as some groups are absent in \uniform sample. Similar results are seen
in other \sasg queries, where \cs, \rl and \cvopt have the maximum
errors 39\%, 
22\% and 
21\% 
for B\ref{query:trips_sasg}; and 14\%, 
34\% and 
8\% 
for AQ\ref{query:openaq_sasg_complex}, respectively.

\begin{table}[t] 
\setlength{\tabcolsep}{2.6pt}
\begin{tabular}{|c|cccc|cccc|}
\hline
         & \multicolumn{4}{c|}{\openaq} & \multicolumn{4}{c|}{\bikes} \\ \cline{2-9} 
         & \sasg &\masg   & \samg	& \mamg		& \sasg   & \masg  & \samg	& \mamg	\\ \hline
\uniform & 21.2 & 19.0  & 12.3 & 10.9	& 14.7  & 9.0  & 24.0 & 20.5 \\
\sps	 & 38.4	& 20.9  & 34.1 & 33.2	& 10.9  & 15.6 & 15.3 & 15.2 \\ 		
\cs      & 2.1 	& 1.1   & 3.2  & 2.3	& 4.8   & 2.6  & 6.9  & 5.2 \\
\rl      & 3.0  & 1.8   & 4.5  & 3.6	& 4.3   & 2.8  & 7.6  & 5.8 \\
\cvopt   & 1.6  & 0.8   & 2.4  & 2.2	& 4.0   & 2.3  & 6.3  & 4.8 \\ \hline
\end{tabular}
\centering
\caption{Percentage average error for different queries, \openaq and \bikes datasets, with $1\%$ and $5\%$ samples, respectively.}
\label{tbl:avg}
\end{table}

Table~\ref{tbl:avg} summarizes the average errors of different
queries. Generally, \cvopt shows the best average error among different methods. 
{\color{red} In some cases, \cvopt has minor improvements for average error, but for maximum error, \cvopt significantly outperforms others.}
That is because \cvopt gives a good representation for \textit{all groups} while others do not guarantee to cover all groups.  The order of other methods changes for different queries, as each method has an advantageous query type, while \cvopt is fairly stable across multiple types of query.  Note that Queries AQ\ref{query:openaq_sasg}, B\ref{query:trips_masg}, and B\ref{query:trips_sasg} have selection predicates, which  are applied
after the sampling is performed. 
{\color{red} We study how the sample can be reused for predicates with different 
selectivities and for different group-by attributes in Section~\ref{exp:sensitivity}.}
{\color{red} 
Since the error of \sps can be very large (maximum error as high as 173\%), we 
exclude it from comparisons in the rest of this section.}
\subsection{Weighted aggregates}
\label{exp:weighted}
\begin{figure}[t]
	\centering
	\begin{subfigure}{\columnwidth}
        \centering
		\includegraphics[width=\graphwidth]{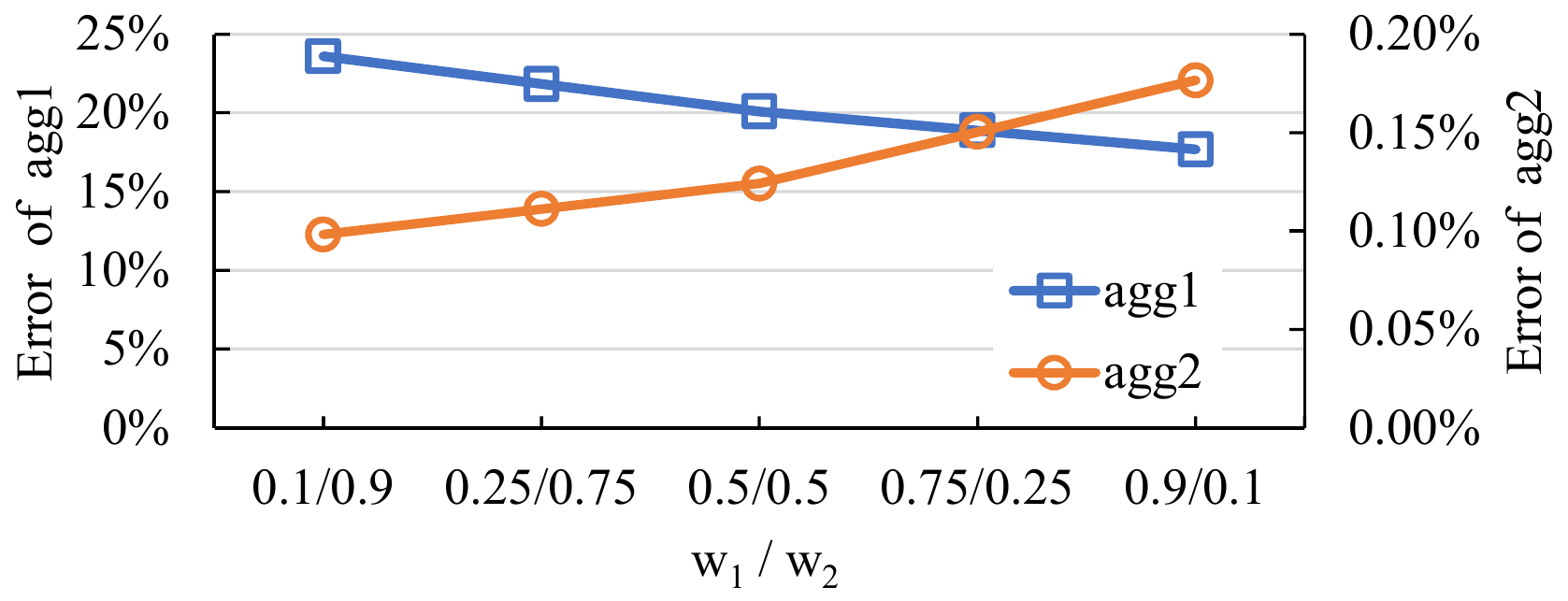}
		\caption{$1\%$ \cvopt sample answers query AQ\ref{query:openaq_masg} with weight settings.}
	\end{subfigure}
	\par\bigskip
	\begin{subfigure}{\columnwidth}
        \centering
		\includegraphics[width=\graphwidth]{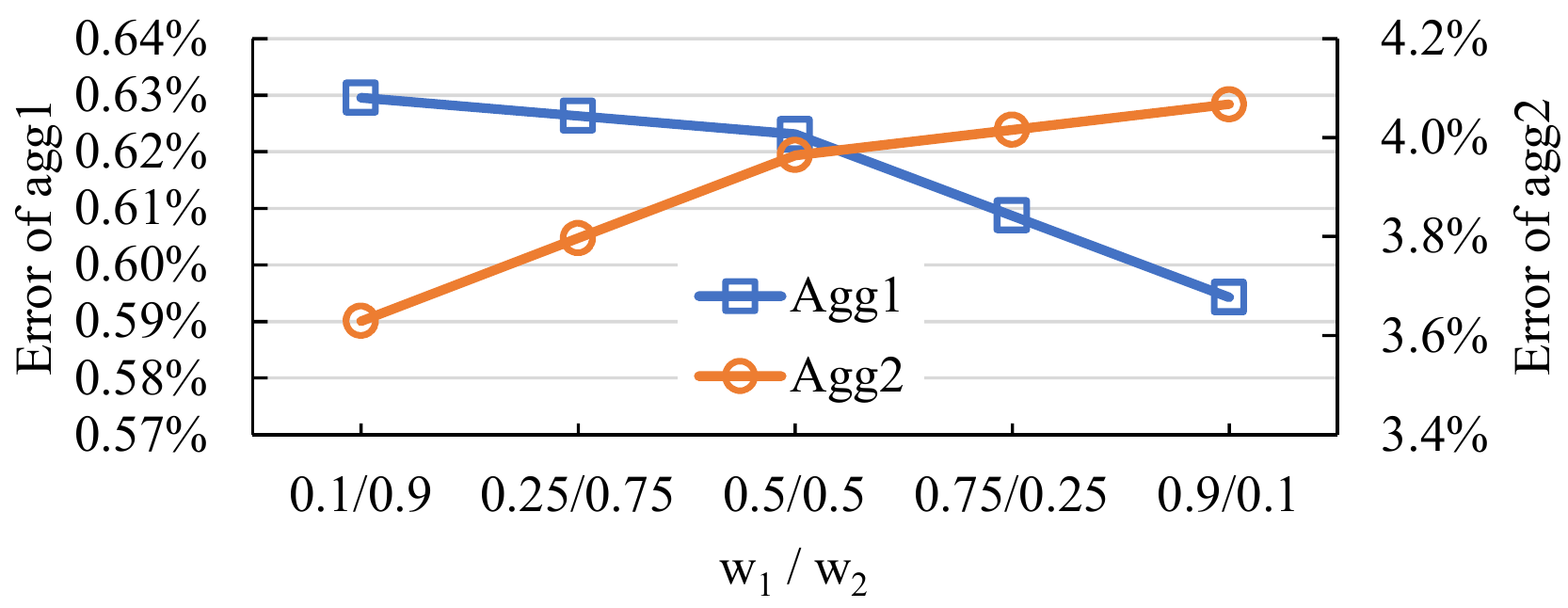}
		\caption{$5\%$ \cvopt sample answers query B\ref{query:trips_masg} with weight settings.}
	\end{subfigure}	
	\caption{Average errors, \cvopt, for different weight settings.}
\label{fig:masg_weighted}
\end{figure}

When multiple aggregates are involved, they are not necessarily
equally important to the user. \cvopt allows the user to assign a
greater weight to one aggregate over the others, leading to an
improved quality for the answer, based on the user's need.  We
conducted experiments with query AQ\ref{query:openaq_masg} and query
B\ref{query:trips_masg}. Each query has two aggregations, {\tt Agg1}
and {\tt Agg2}, with the weights denoted by $w_1$ and $w_2$,
respectively. We use \cvopt to draw 3 samples with different weighting
profiles:
$(w_1, w_2) = \big\{(0.1, 0.9), (0.5, 0.5), (0.9, 0.1)\big\}$, as the
user favors {\tt Agg2} over {\tt Agg1} in the first case and favors \textcolor{red}{
{\tt Agg1}} in the third case. The second case is equal to default,
non-weighted setting. Results are presented in
figure~\ref{fig:masg_weighted}. From the left to the right side, as
$w_1$ increases and $w_2$ decreases, the average error of {\tt Agg1}
decreases and of {\tt Agg2} increases. The results in both queries
show \cvopt's ability to create a sample that better fits the user's
priority. While previous heuristic works cannot systematically support weighted aggregates, we find this feature is practically useful in calibrating the sample.

\subsection{Sample's usability}
\label{exp:sensitivity}

\begin{figure}[t]
	\centering
    \begin{subfigure}{\columnwidth}
        \centering
		\includegraphics[width=\graphwidth]{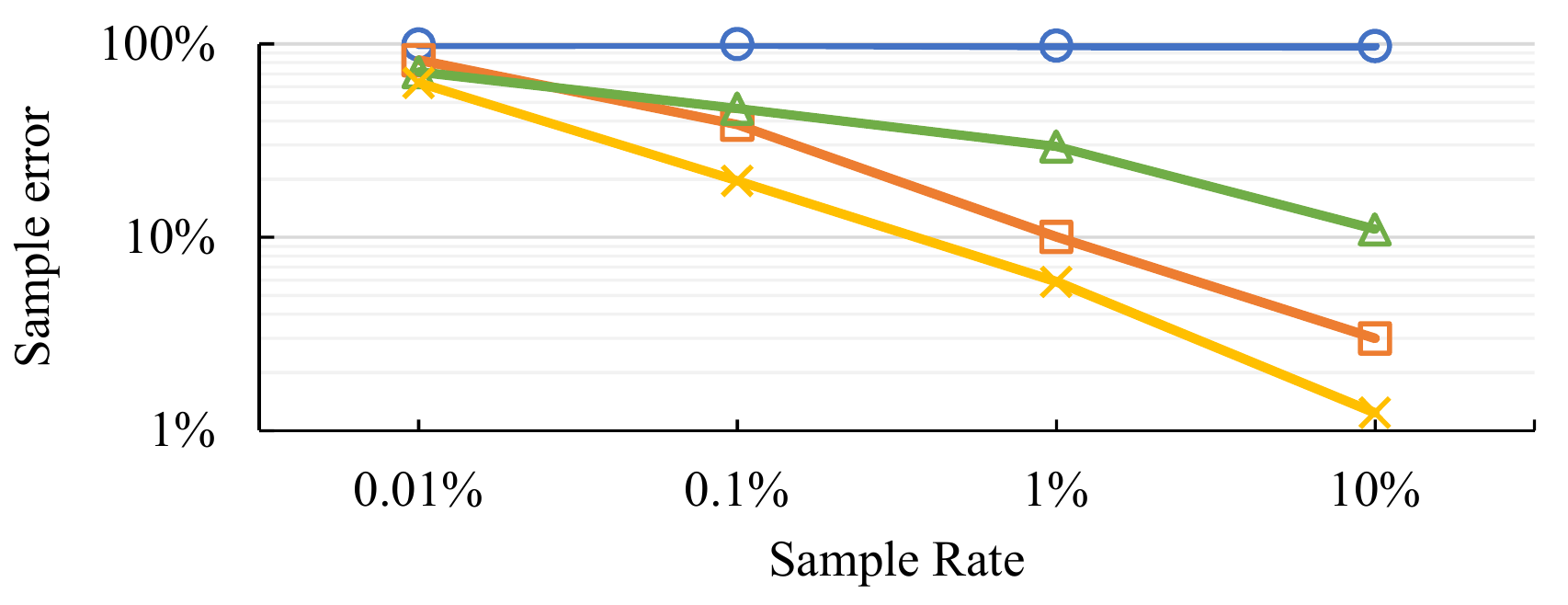}
		\caption{\masg query AQ\ref{query:openaq_masg} answered by samples with various sample rates.}
    \end{subfigure}
    \par\bigskip
	\begin{subfigure}{\columnwidth}
        \centering
		\includegraphics[width=\graphwidth]{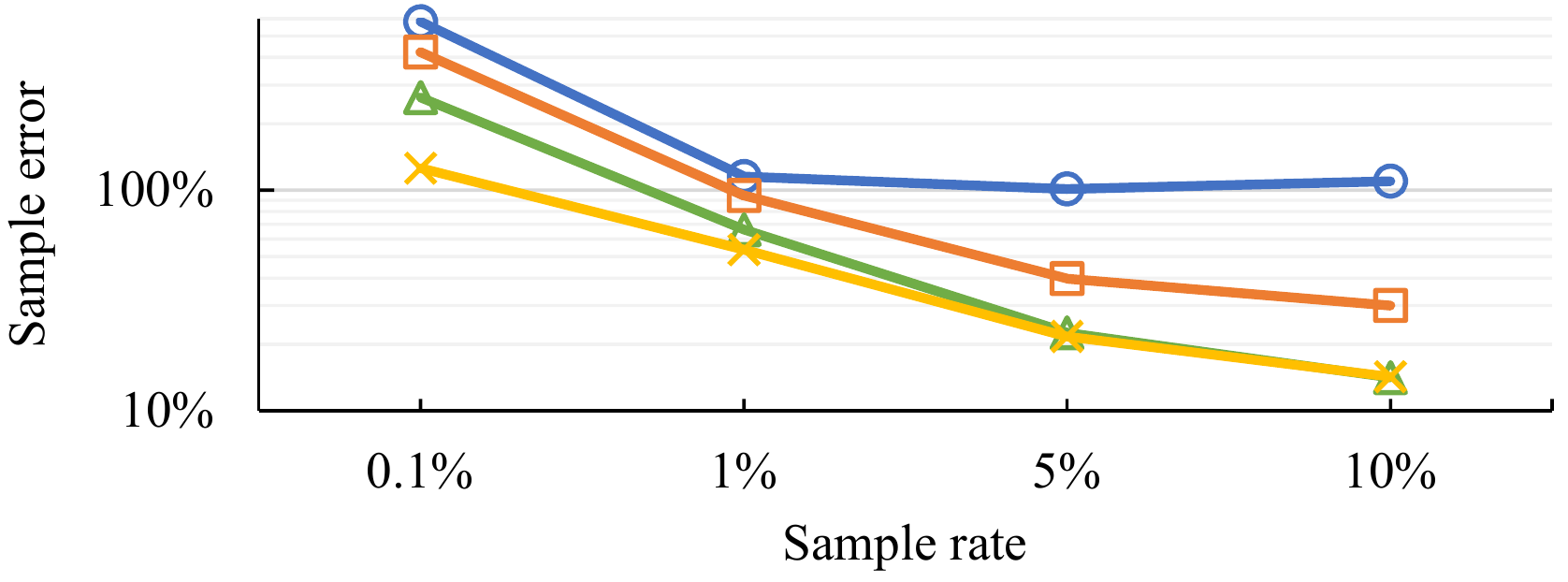}
		\caption{\sasg query B\ref{query:trips_sasg} answered by samples with various sample rates.}
	\end{subfigure}	
    \includegraphics[width=\graphwidth]{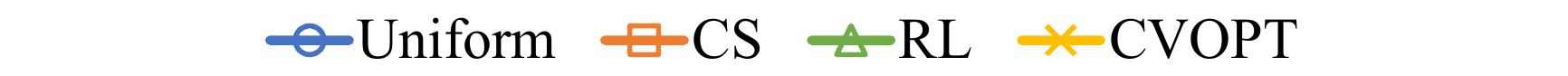}
    \caption{Sensitivity of maximum error to sample size.}
    \label{fig:sample_rate}
\end{figure}

{\color{red}{On Sample Rate:} }
It is well known that a higher sample rate generally improves the
sample's accuracy in query processing. Nevertheless, we compare
the accuracy of the
different methods for queries AQ\ref{query:openaq_masg} and B\ref{query:trips_sasg}
under different sample rates. 
Figure~\ref{fig:sample_rate} shows that \cvopt outperforms its peers
at nearly all sample rates of the study.

\begin{figure}[t]
	\begin{subfigure}{\columnwidth}
    	\centering
		\includegraphics[width=\graphwidth]{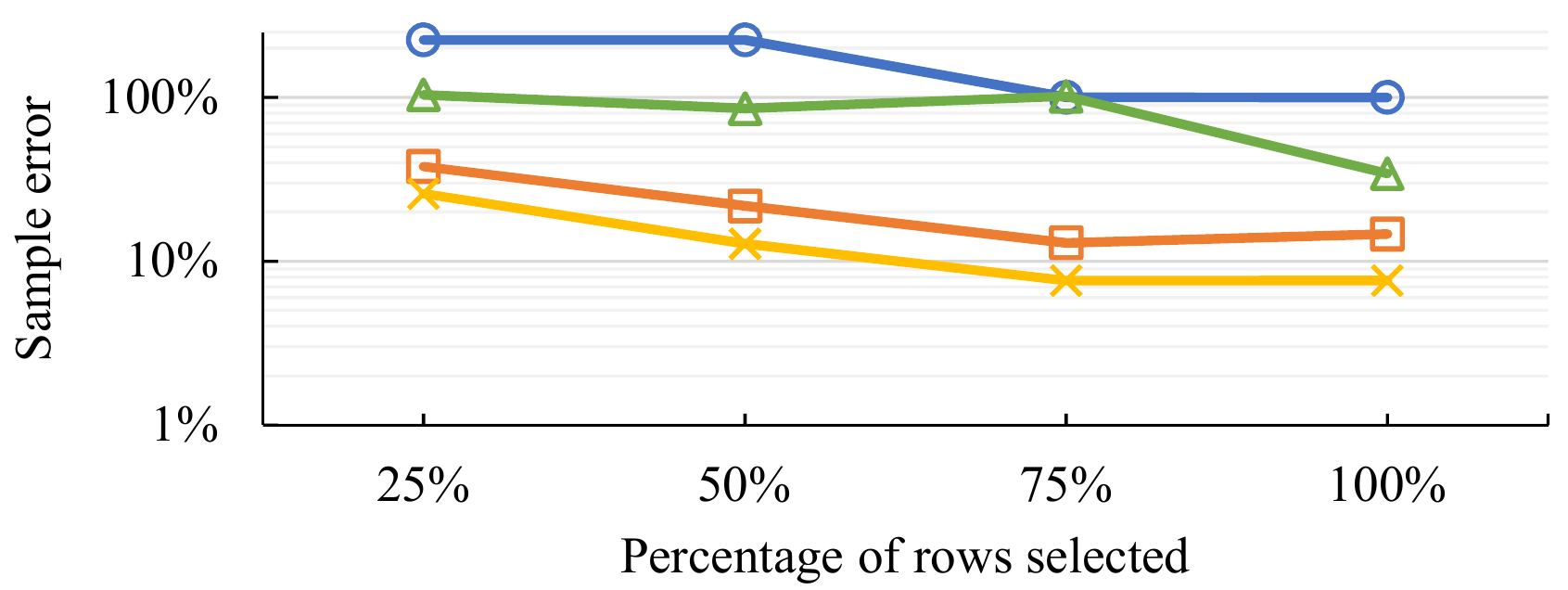}
		\caption{\sasg queries with different predicates on \openaq.}
	\end{subfigure}	
	\par\bigskip
	\begin{subfigure}{\columnwidth}
    	\centering
		\includegraphics[width=\graphwidth]{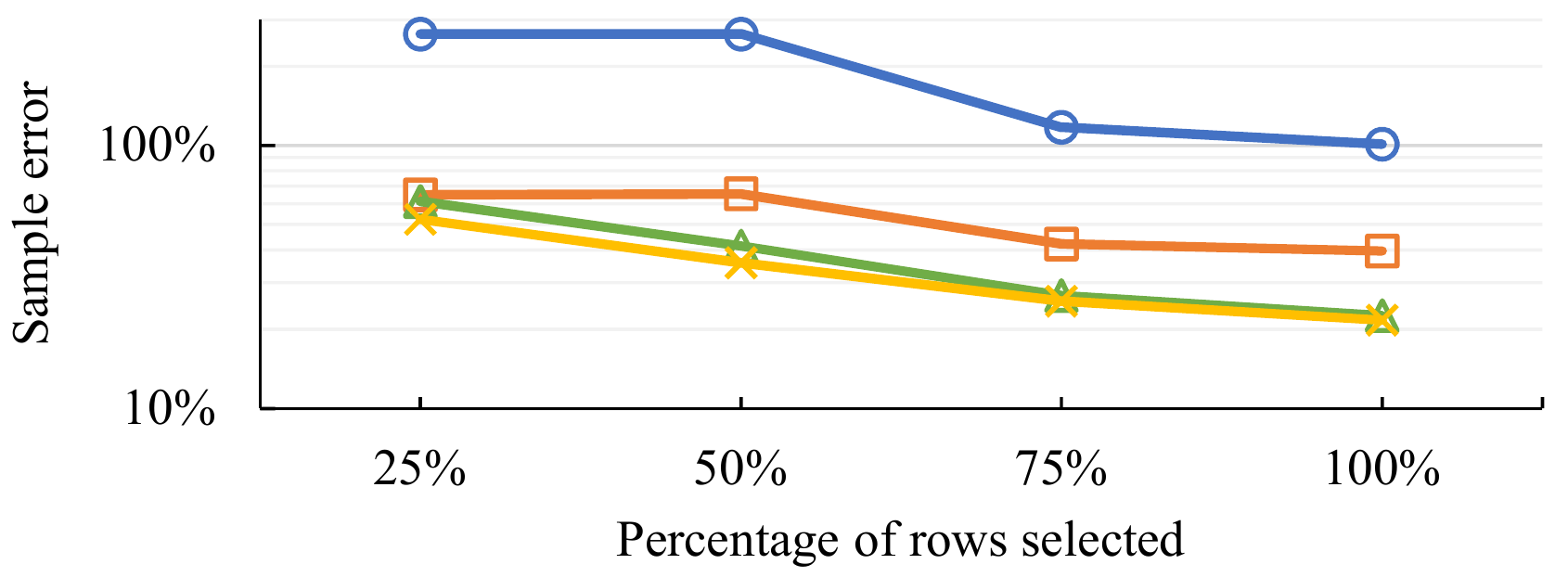}
		\caption{\sasg queries with different predicates on \bikes.}
		\label{fig:selectivity_openaq}
	\end{subfigure}	
    \includegraphics[width=\graphwidth]{exp/line_lable.pdf}
	\caption{\color{red} Maximum error of multiple queries with different predicates answered
          by one materialized sample, showing
the reusability of the sample}
	\label{fig:selectivity}
\end{figure}

\medskip 

{\color{red}{On Predicate Selectivity and Group-by Attributes:} }
Queries commonly come with
selection predicates, \ie the {\tt WHERE} clause in a SQL
statement. Since samples are pre-constructed, the same sample is used no matter what the
selection predicate is. {\color{red}We study the sample's usability on 
the selectivity of the predicate.} Query AQ\ref{query:openaq_sasg} has
a predicate: {\tt \textbf{WHERE HOUR}(local\_time) \textbf{BETWEEN}
  <low> \textbf{AND} <high>}.  By changing the filter
conditions, 
{\color{red} we generated 4 queries AQ\ref{query:openaq_sasg}.a,
  AQ\ref{query:openaq_sasg}.b, AQ\ref{query:openaq_sasg}.c and
  AQ\ref{query:openaq_sasg}, with selectivity of $25\%$,
  $50\%$, $75\%$ and $100\%$, respectively. 
All
  queries are answered by the materialized sample optimized for query
  AQ\ref{query:openaq_sasg}. 
}
\remove{
  To show the robustness of \cvopt, we
  applied predicate on the time of measurement, \texttt{local\_time},
  which is highly correlated with the \texttt{value} measurement. \todo{what
    does this sentence mean? Two measurements here and they are
    correlated ? what the two measurements? What does that mean by correlated ?} The $L_1$
  distance between the distributions of the grouped frequency
  \todo{what exactly is a group frequency? What is the def of the L1 distance of
    them? What is L1 distance is trying to show? How sensitive is it?
    Why the L1 distance shall be changing dramatically when the
    ``local time'' filter is changed? Give a simple example? Any
    insights why the sample still works well when this filter
    changes? The sample shall not be work well for an arbitrary
    predicate/filter change.} with and
  without predicate is covariated with predicate selectivity, \ie
  0.0144, 0.0356 and 0.1024 for AQ\ref{query:openaq_sasg}.c,
  AQ\ref{query:openaq_sasg}.b, AQ\ref{query:openaq_sasg}.a,
  respectively. It shows that the predicate affects not only the
  number of row selected but also the distribution of the data.
}
{\color{red}
  Similarly, we have queries B\ref{query:trips_sasg}.a,
  B\ref{query:trips_sasg}.b, B\ref{query:trips_sasg}.c, and
  B\ref{query:trips_sasg}, with controllable predicate parameters to
  table \bikes.  
} Figure~\ref{fig:selectivity} shows the results. Although
  the samples are able to serve different queries, we observe
  the effect of selectivity upon the accuracy.  Overall, the greater the
selectivity, the lesser is the error due to sampling. For each
predicate query, \cvopt has a lower error than \cs and \rl.\looseness=-1
{\color{red}
\begin{query}[b]
\begin{lstlisting}
SELECT	country, parameter, unit,
	AVG(value) average
FROM OpenAQ	WHERE latitude > 0
GROUP BY country, parameter, unit
\end{lstlisting}
\caption{\color{red}Average measurement for each parameter of the countries in northern hemisphere.}
\label{query:openaq_sasg_2}
\end{query}
\begin{query}[b]
\begin{lstlisting}
SELECT
    parameter, unit,
COUNT(IF(value > 0.5, 1, 0)) count
FROM {input_table}	WHERE country = "VN"
GROUP BY parameter, unit
\end{lstlisting}
\caption{\color{red}Count the number of times the measurement of each parameter is higher than 0.5 in Vietnam.}
\label{query:openaq_sasg_3}
\end{query}
\begin{table}[t]
	\centering
	\begin{tabular}{|c|cccccc|}
		\hline
		& AQ\ref{query:openaq_sasg}  & AQ\ref{query:openaq_sasg}.a & AQ\ref{query:openaq_sasg}.b & AQ\ref{query:openaq_sasg}.c & AQ\ref{query:openaq_sasg_2}  & AQ\ref{query:openaq_sasg_3}   \\ \hline
		\uniform & 98.4 & 21.0  & 21.4  & 18.0  & 99.6 & 100.0 \\
		\cs      & 2.5  & 5.8   & 2.9   & 2.8   & 3.9  & 0.9   \\
		\rl      & 5.4  & 9.5   & 6.9   & 5.6   & 4.3  & 3.5   \\
		\cvopt    & 1.5  & 4.4   & 2.4   & 1.9   & 2.3  & 0.8   \\ \hline
	\end{tabular}
	\caption{\color{red} Average error of multiple queries
          answered by one materialized sample, showing the reusability
          of the sample.}
	\label{tbl:multiple_queries}
\end{table}
%
Table~\ref{tbl:multiple_queries} shows the average error of a group of
six different queries using the materialized sample optimized for
AQ\ref{query:openaq_sasg}.  Note that 
 AQ\ref{query:openaq_sasg_2} and
AQ\ref{query:openaq_sasg_3} use different \texttt{\bf WHERE} clauses
that are also different from those used by 
AQ\ref{query:openaq_sasg} and AQ\ref{query:openaq_sasg}.*. 
AQ\ref{query:openaq_sasg_3} also uses a different set of
\texttt{\bf GROUP-BY} attributes.  For all six queries, \cvopt
performs well with good accuracy and is consistently better than other methods.  }

\subsection{Multiple Group-by Query}
\label{exp:samg}
\begin{query}[b]
\renewcommand{\ttdefault}{pcr}
\begin{lstlisting}
SELECT country, parameter, SUM(value)
FROM OpenAQ
GROUP BY country, parameter WITH CUBE
\end{lstlisting}
\caption{Single aggregate, multiple group-by, \openaq.}
\label{query:openaq_samg}
\end{query}
\begin{query2}[b]
\renewcommand{\ttdefault}{pcr}
\begin{lstlisting}
SELECT  from_station_id, year, 
        SUM(trip_duration)
FROM Bikes	WHERE age > 0
GROUP BY from_station_id, year WITH CUBE
\end{lstlisting}
\caption{Single aggregate, multiple group-by, \bikes}
\label{query:trips_samg}
\end{query2}

\begin{query}[b]
\renewcommand{\ttdefault}{pcr}
\begin{lstlisting}
SELECT country, parameter, 
       SUM(value), SUM(latitude)
FROM OpenAQ
GROUP BY country, parameter WITH CUBE
\end{lstlisting}
\caption{Multiple aggregate, multiple group-by, \openaq.}
\label{query:openaq_mamg}
\end{query}
\begin{query2}[b]
\renewcommand{\ttdefault}{pcr}
\begin{lstlisting}
SELECT  from_station_id, year, 
        SUM(trip_duration), SUM(age)
FROM Bikes
GROUP BY from_station_id, year WITH CUBE
\end{lstlisting}
\caption{Multiple aggregate, multiple group-by, \bikes}
\label{query:trips_mamg}
\end{query2}

\texttt{WITH CUBE} is an extension of group-by clause in SQL that generates multiple grouping sets of given attributes in a single query. For example, \texttt{\textbf{GROUP BY} A, B \textbf{WITH CUBE}} will generate grouping sets of \texttt{(A, B)}, \texttt{(A)}, \texttt{(B)}, and \texttt{()}. \texttt{WITH CUBE} is powerful in helping user to easily and efficiently compute aggregations with a large number of combinations of group-by.
We conduct experiments with \samg queries AQ\ref{query:openaq_samg} and B\ref{query:trips_samg} and \mamg queries AQ\ref{query:openaq_mamg} and B\ref{query:trips_mamg}. The queries have the grouping sets of two attributes for multiple group-by. \openaq has 38 countries and 7 parameters, so cube with these two attributes will generate upto 312 groups. \bikes has 619 stations and 3 years of collection, and therefore cube with \textit{from\_station\_id} and \textit{year} leads upto 2480 groups.

\begin{figure}[t]
	\centering
	\includegraphics[width=\graphwidth]{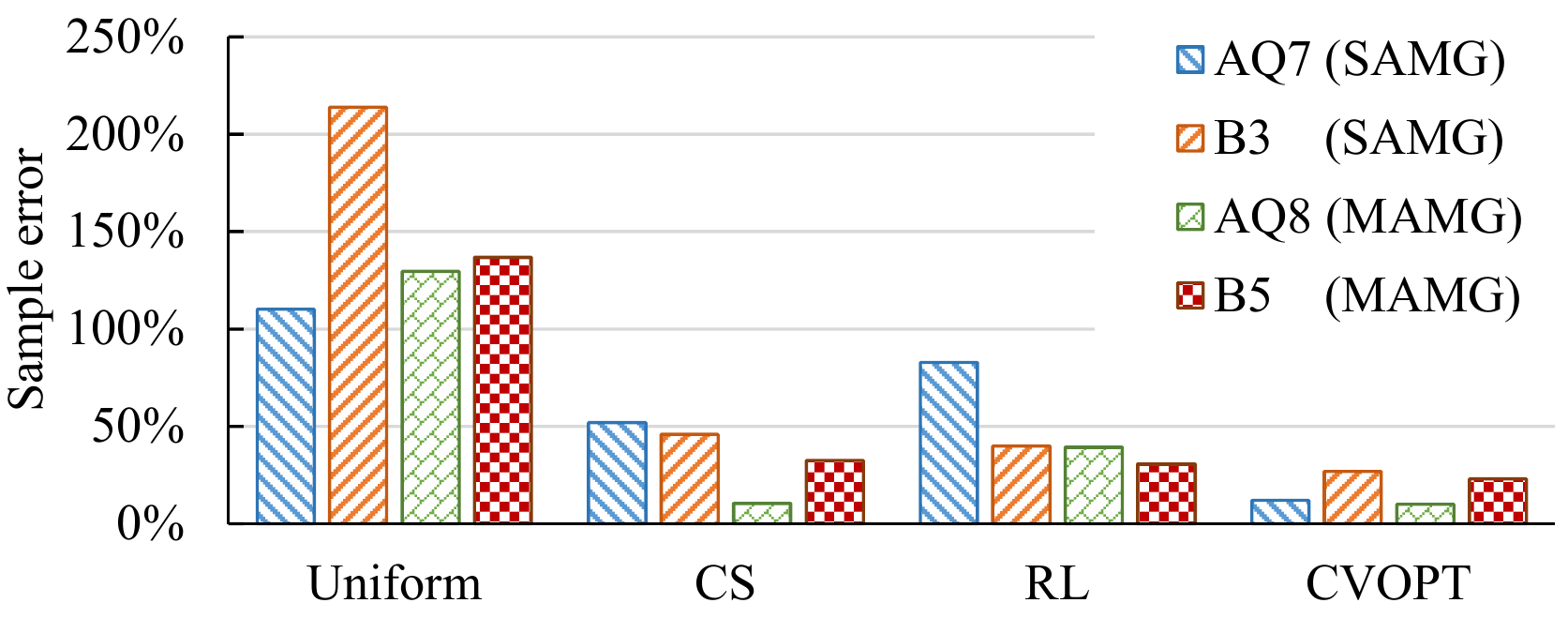}
	\caption{Maximum error of {\tt CUBE} group-by queries.} 
	\label{fig:samg_per}
\end{figure}

\remove{
\begin{figure}[t]
	\centering
	\includegraphics[width=\graphwidth]{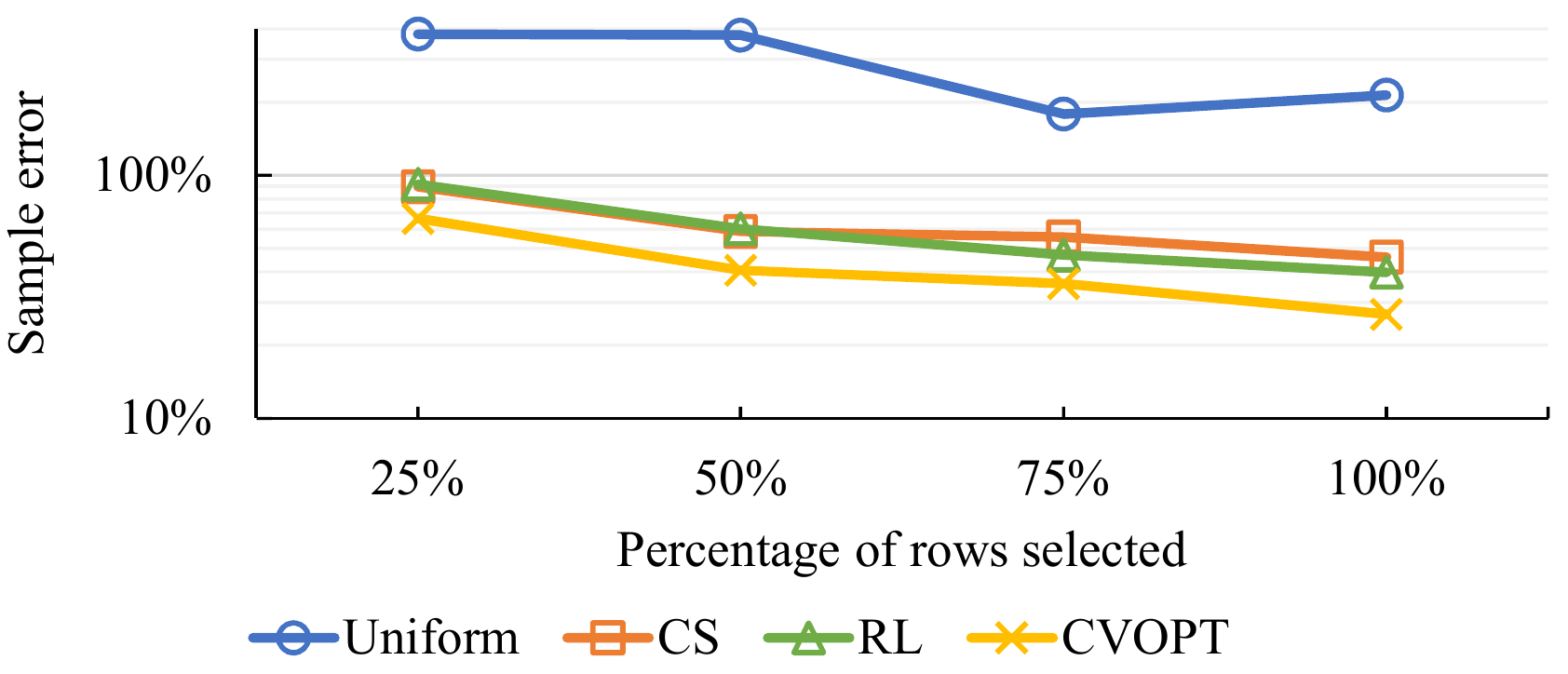}	
	\caption{Impact of predicate selectivity, ${99}^{th}$ percentile error, query B\ref{query:trips_samg}}
	\label{fig:samg_sel}
\end{figure}

\begin{figure}[t]
	\centering
	\includegraphics[width=\graphwidth]{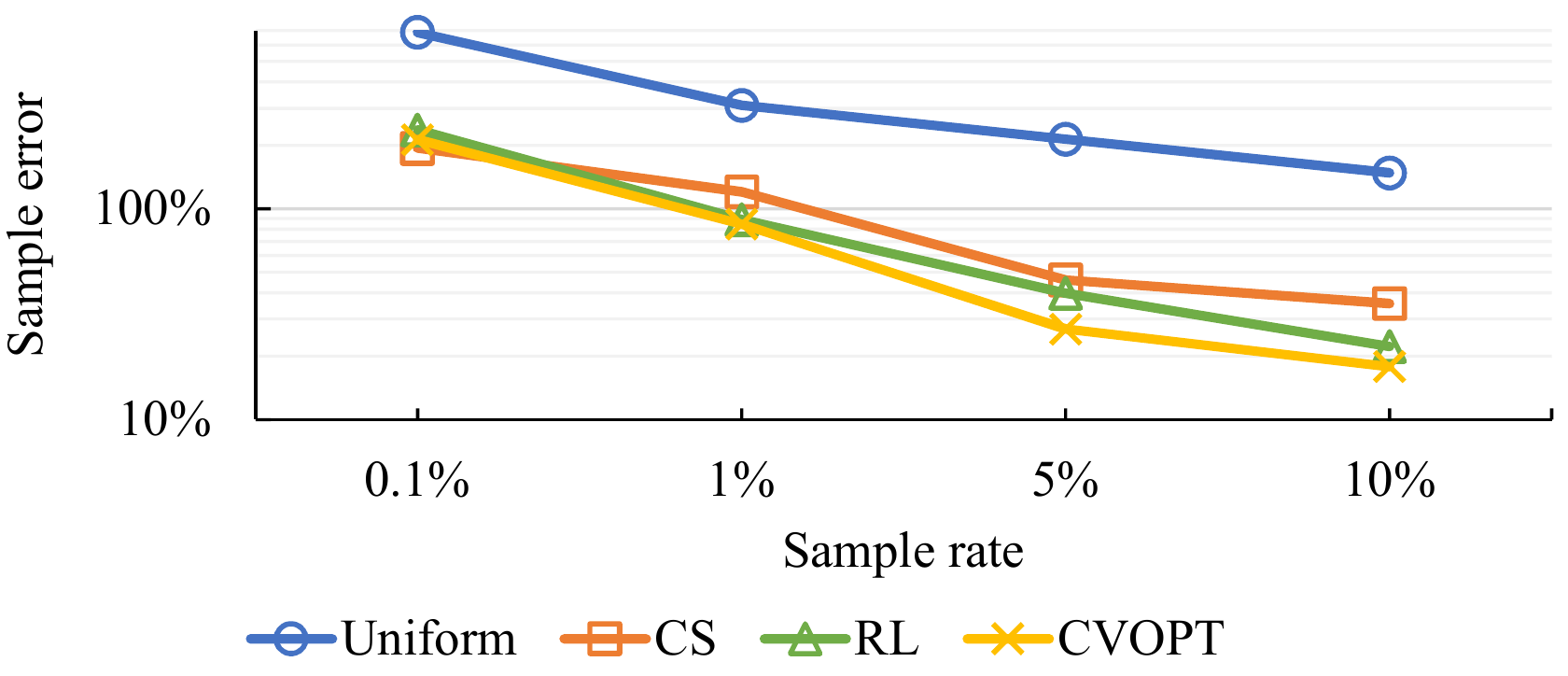}
	\caption{Impact of sample rate, ${99}^{th}$ percentile error, query B\ref{query:trips_samg}}
	\label{fig:samg_sample_rate}
\end{figure}
}

All methods \rl, \cs, and \cvopt, can sample in the presence of multiple groupings; for \cs this is the scaled congressional sampling method and for \rl it is the hierarchical partitioning method. Both \cs and \rl adopt a heuristic approach, which we implemented as described in their algorithm.
Accuracies of different samplers are shown in Figure~\ref{fig:samg_per}. We note that \cvopt performs significantly better than \uniform and \rl and is consistently better than \cs.




\remove{
We also conducted experiments to measure the sensitivity of sample quality to the sample rate and predicate selectivity, results shown in
Figures~\ref{fig:samg_sel} and~\ref{fig:samg_sample_rate}. 
\todo{while all other figs focus on max error, why do these two figs focus on 99th error?}
}

{\color{red}
\subsection{CPU Time}
\label{exp:time}

%

\begin{table}[]
\centering
\begin{tabular}{|c|c|c|c|c|}
\hline
\multirow{2}{*}{} & \multicolumn{2}{c|}{\openaq (40~GB)} & \multicolumn{2}{c|}{\openaqtb~(1~TB)} \\ \cline{2-5} 
                          & Precomputed     & Query     & Precomputed      & Query        \\ \hline
\texttt{Full Data} & -               & 2881.68   & -           & 60565.49      \\
\uniform             & 913.95      & 40.31       & 9060.53	  &  294.40     \\
\sps                   & 2309.99     & 44.92       & 58119.91  &  821.46    \\
\cs                		& 3854.61     & 55.87       & 87036.11  &  1091.60    \\
\rl                		  & 4311.36    & 54.43       & 99365.67  &  952.37    \\
\cvopt                & 4263.07    & 59.99       & 98081.49  &  1095.18   \\ \hline
\end{tabular}
\caption{\color{red} The sum of the CPU time (in seconds) at all nodes for sample precomputing and
  query processing with 1\% sample for query
  AQ\ref{query:openaq_masg_complex} over \openaq and \openaqtb.}
\label{tbl:time}
\end{table}

We measure the CPU time cost for sample precomputation and query
processing, as the sum of the CPU cost at
all nodes in a Hadoop cluster we used. The cluster includes a master
node and 3 slave nodes. Each node has 2.70~GHz~x4 CPU with 10~GB
memory. For timing purposes, we generate a 1~TB dataset, \openaqtb, by duplicating \openaq 25 times.
\remove{
The main messages are: 1) The overhead caused by sample
precomputation is negligible, compared to the time we save from
sample-based query processing. 2) \cvopt has nearly the same
performance as its best competitors \cs and \rl, on
both sample precomputation and query response, while maintaining much
better query response accuracy.
}


Table~\ref{tbl:time} shows the CPU time for precomputing  the
$1\%$ samples and query processing to answering
query AQ\ref{query:openaq_masg_complex}. 
Due to the two scans required, stratified random samples, \sps,
\cs, \rl and \cvopt, are more expensive than \uniform, which requires only 
a single scan through data. 

Query processing using samples could reduce CPU usage by $50$ to $300$ times. 
\cs and \sps have slightly less CPU time. $\rl$ has comparable CPU
time cost. Since \uniform missed a number of groups, its query time is smaller.
\cvopt consumed only about $2\%$ of the CPU compared to query over the base table, reducing
the CPU time from $0.8$ hour to less than a minute for the 40GB
\openaq, and from $17$ hours to $18$ minutes for the 1TB
\openaqtb.  

For both data sets, the precomputation time of \cvopt is about $1.5$x
the cost of running the original query on the full data, indicating
that with only \emph{two} queries, the precomputation time can be amortized to be cheaper than running the
original queries on the full data.
}

\subsection{Experiments with \cvoptinf}
\label{exp:inf}

\begin{figure}[t]

	\centering
	\includegraphics[width=\graphwidth]{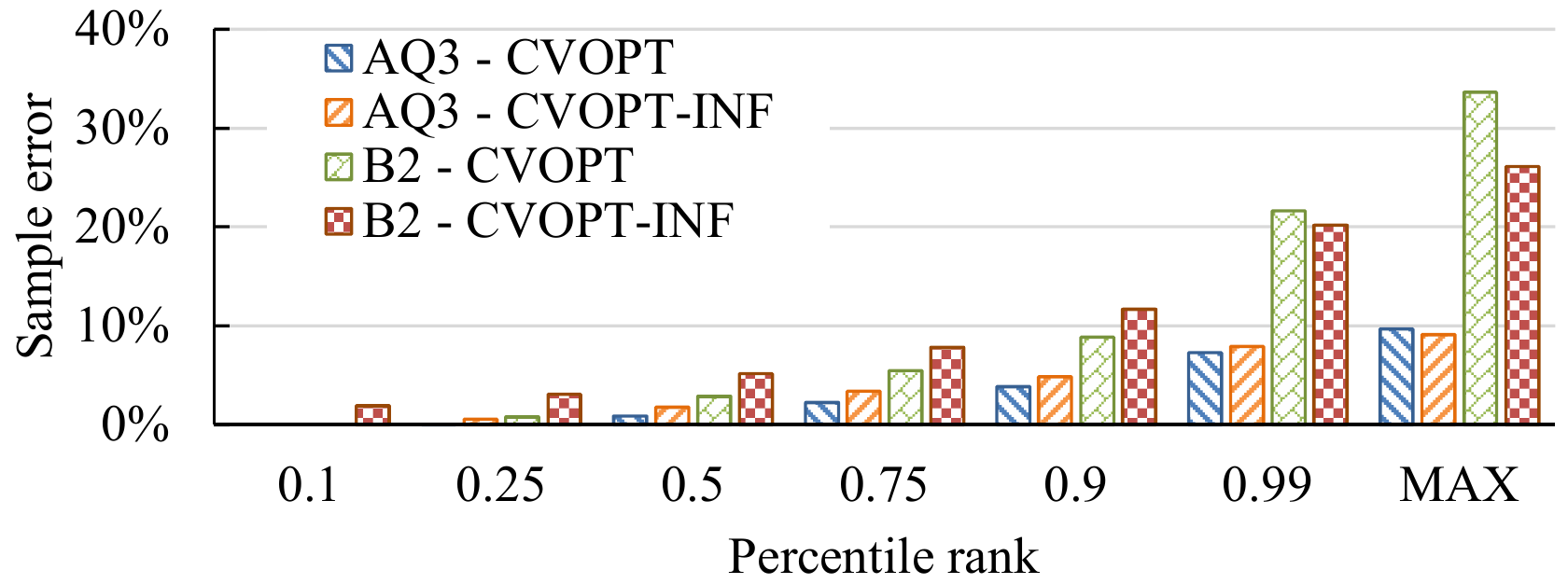}
	\caption{Comparison accuracy from \cvopt and \cvoptinf for \sasg queries AQ\ref{query:openaq_sasg} and B\ref{query:trips_sasg}.}	
	\label{fig:inf_per}
\end{figure}

Our experiments show that \cvopt, which optimizes for the $\ell_2$ norm of the CVs, leads to smaller errors at higher percentiles than \cs and \rl. We now consider \cvoptinf, which minimizes for the maximum of CVs of all estimates ($\ell_{\infty}$ norm). Our results on the accuracy of \cvoptinf on queries AQ\ref{query:openaq_sasg} and B\ref{query:trips_sasg} are shown in Figure~\ref{fig:inf_per}. Consistent with theoretical predictions, \cvoptinf has a lower maximum error than \cvopt on both queries. At the same time, \cvoptinf has a worse error than \cvopt at the $90^{th}$ percentile and below. Overall, this shows that \cvoptinf can be considered when the user is particularly sensitive to the maximum error across all groups. Otherwise, \cvopt with $\ell_2$ optimization provides robust estimate for a large fraction of the groups, with a small error across a wide range of percentiles.
	

\section{Other Related Work}
\label{sec:other-related}
Random sampling has been widely used in approximate query processing, for both static~\cite{Chaudhuri2017, KeYi17, SKT2012, BCD-SIGMOD2003, Coch77:book,Lohr-book2009, YTill-book1997} and streaming data~\cite{HentschelHT18, HAAS2016, Zhang:2016, Ahmed2017}. Uniform sampling~is~simple and can be implemented efficiently (e.g. using reservoir sampling~\cite{Vitter-sampling-focs83}), but does not produce good estimators for groups of low volume or large variance. 
Ganti et al.~\cite{GLR-VLDB2000} address low selectivity queries using workload-based sampling, such that a group with a low selectivity is sampled as long as the workload includes queries involving that group. Different techniques have been used in combination with sampling, such as indexing~\cite{sampleandseek, Wang:2015, CDMN01} or aggregate precomputation~\cite{AQP}.

Chaudhuri et al.~\cite{CDN-TODS2007} formulate approximate query processing as an optimization problem. Their goal is to minimize the $\ell_2$ norm of the relative errors of all queries in a given workload. Their approach to group-by queries is to treat every group derived from every group-by query in the workload as a separate query. In doing so, their technique does not handle overlap between samples suited to different groups. In contrast, our framework considers the overlaps and interconnections between different group-by queries in its optimization. 

Kandula et al.~\cite{KSVOGCD-SIGMOD2016} consider queries that require multiple passes through the data, and use random sampling in the first pass to speed up subsequent query processing. This work can be viewed as query-time sampling while ours considers sampling before the full query is seen. Further, \cite{KSVOGCD-SIGMOD2016} does not provide error guarantees for group-by queries, like we consider here. A recent work~\cite{edbt19} has considered stratified sampling on streaming and stored data, addressing the case of full-table queries using an optimization framework. This work does not apply to group-by queries.

	\section{Conclusion}
	We presented \cvopt, a framework for sampling from a database to accurately answer group-by queries with a provable guarantee on the quality of allocation, measured in terms of the $\ell_2$ (or $\ell_{\infty}$) norm of the coefficient of variation of the different per-group estimates. Our framework is quite general, and can be applied anytime it is needed to optimize jointly for multiple estimates. Choosing the $\ell_2$ of the CVs, larger errors are penalized heavily, which leads to an allocation where errors of different groups are concentrated around the mean. Choosing the $\ell_{\infty}$ of the CVs leads to a lower value of the maximum error than the $\ell_2$, at the cost of a slightly larger mean and median error. There are many avenues for further research, including (1)~incorporating joins into the sampling framework (2)~exploring $\ell_{p}$ norms for values of $p$ other than $2,\infty$, (3)~handling streaming data.
	
	\balance
	
	{\scriptsize
		\bibliographystyle{abbrv}
		\bibliography{ms}
	}
	\balance
	
	
\end{document}